\documentclass[letterpaper]{article} 
\usepackage{AAAI_2027/aaai2027}  
\nocopyright

\usepackage[hyphens]{url}  
\usepackage{graphicx} 
\usepackage{natbib}  
\usepackage{caption} 
\usepackage{booktabs}

\usepackage{amsmath}
\usepackage{amsthm}
\usepackage[textsize=scriptsize,textwidth=2cm,color=green!20]{todonotes}
\usepackage{dsfont}
\usepackage{subcaption}
\usepackage{multicol}
\usepackage{cleveref}
\usepackage{adjustbox}
\usepackage{comment}
\usepackage{bbding}
\usepackage{units}
\usepackage{centernot}
\usepackage{multirow}
 
\usepackage[ruled]{algorithm2e} 

\SetAlFnt{\small}
\SetAlCapFnt{\small}
\SetAlCapNameFnt{\small}
\SetAlCapHSkip{0pt}
\IncMargin{-\parindent}

\makeatletter
\newtheorem*{rep@theorem}{\rep@title}
\newcommand{\newreptheorem}[2]{%
\newenvironment{rep#1}[1]{%
 \def\rep@title{#2 \ref{##1}}%
 \begin{rep@theorem}}%
 {\end{rep@theorem}}}
\makeatother

\newtheorem{theorem}{Theorem}
\newreptheorem{theorem}{Theorem}
\newtheorem{definition}[theorem]{Definition}
\newtheorem{proposition}[theorem]{Proposition}
\newreptheorem{proposition}{Proposition}
\newtheorem{observation}[theorem]{Observation}
\newreptheorem{observation}{Observation}
\newtheorem{lemma}[theorem]{Lemma}
\newreptheorem{lemma}{Lemma}

\newtheorem{remark}[theorem]{Remark}

\newreptheorem{corollary}{Corollary}

\RequirePackage{tikz-cd}
\RequirePackage{amssymb}
\usetikzlibrary{calc}
\usetikzlibrary{decorations.pathmorphing}
\usetikzlibrary{spath3}

\tikzset{curve/.style={settings={#1},to path={(\tikztostart)
    .. controls ($(\tikztostart)!\pv{pos}!(\tikztotarget)!\pv{height}!270:(\tikztotarget)$)
    and ($(\tikztostart)!1-\pv{pos}!(\tikztotarget)!\pv{height}!270:(\tikztotarget)$)
    .. (\tikztotarget)\tikztonodes}},
    settings/.code={\tikzset{quiver/.cd,#1}
        \def\pv##1{\pgfkeysvalueof{/tikz/quiver/##1}}},
    quiver/.cd,pos/.initial=0.35,height/.initial=0}

\tikzset{between/.style n args={2}{/tikz/execute at end to={
    \tikzset{spath/split at keep middle={current}{#1}{#2}}
}}}

\tikzset{tail reversed/.code={\pgfsetarrowsstart{tikzcd to}}}
\tikzset{2tail/.code={\pgfsetarrowsstart{Implies[reversed]}}}
\tikzset{2tail reversed/.code={\pgfsetarrowsstart{Implies}}}
\tikzset{no body/.style={/tikz/dash pattern=on 0 off 1mm}}

\newcommand{\davide}[1]{\textcolor{teal}{[Davide: {#1}]}}

\newcommand{\e}{\text{e}}

\newcommand{\N}{\mathbb{N}}
\newcommand{\gr}{\textsc{gr}}
\newcommand{\opt}{\text{OPT}}

\newcommand{\SSJR}{\text{SSJR}}
\newcommand{\SJR}{\text{SJR}}
\newcommand{\NASH}{\text{NASH}}

\title{Diverse Representation in Approval-Based Committee Voting}
\author{
Julian Chingoma\textsuperscript{\rm 1},
Davide Grossi\textsuperscript{\rm 1},
Feline Lindeboom\textsuperscript{\rm 1}\corresponding,
Jan Maly\textsuperscript{\rm 2}
}
\affiliations{
    \textsuperscript{\rm 1}Rijksuniversiteit Groningen\\
    \textsuperscript{\rm 2}WU Wien\\

}

\begin{document}
\maketitle

\begin{abstract}
The study of approval-based committee voting has so far focused predominantly on \emph{proportional} representation. The canonical notion of \emph{diverse} representation, based on the Chamberlin-Courant score, counts the number of \emph{voters} with at least one representative in the committee, making it an individualistic notion. We develop a more comprehensive theory that instead requires the representation of many \emph{groups} of voters, grounding it in the justified representation (JR) axiom, which we strengthen in two directions. First, we study the existing axioms of Strong JR (SJR) and Semi-Strong JR (SSJR), which consider the same cohesive groups as JR but demand stricter representation. We show that neither can be optimized efficiently on general domains (unless P = NP), but both can be on the Candidate Interval and Voter Interval domain. Second, to capture the unique and defining opinions of a group, we introduce Distinctive Representation (DR) and its local optimization variant, Local DR. We address satisfiability and computation time, and show (Local) DR to be distinct from known proportionality and diversity axioms. Experiments on real-world and synthetic data show that Local DR performs well on multiple empirical measures of diversity.
\end{abstract}




\section{Introduction}
Selecting a set of comments, views or alternatives that represent a large population of voters or users is a key challenge in the pursuit of making social media, AI alignment or online deliberation fairer and more democratic. How to interpret the term `representation' is the first major design decision one has to make in this endeavour. One option is to aim for \emph{proportional} representation, a direction that the study of approval-based committee voting, the most well-studied form of multi-winner voting \cite{faliszewski2017multiwinner}, has focused on nearly exclusively
\cite{lackner_multi-winner_2023}. 
However, in many applications such as the ones mentioned above, one can argue  that it is more important to represent the {\em diversity} of opinions held by users. First, diverse outcomes guarantee higher inclusivity, an attribute one could desire in and of itself. Second, an inclusive outcome can show more participants that their input is being taken into account. This increases people's trust in the system at hand \cite{mikhaylovskaya2024enhancing}, which motivates them to also participate in future instances  \cite{eder_political_2015}. Third, in the context of deliberation summarization, a diverse summary gives a broad, informative view of `the group's' opinion, which can be especially valuable when the number of input opinions is very large. Fourth, cognitive diversity is argued, e.g., by \citet{landemore_democratic_2017}, to also be of epistemic value from a crowd-wisdom perspective. 

The canonical notion of diversity in approval-based committee voting is based on the Chamberlin-Courant (CC) score \cite{chamberlin_representative_1983}. The CC score measures diversity by counting the number of voters that have some form of representation, making it an individualistic notion. In the real world, people are part of communities, so we desire a diversity notion which takes that into account, representing voters as part of groups. For this, we turn to the proportionality literature, where representation of groups has been studied extensively \cite{lackner_multi-winner_2023}, focusing on the representation of groups that exhibit some form of `cohesiveness'. This grants larger groups more representatives, while not ignoring smaller groups. When diversity is the goal, however, we should \emph{not} assign more candidates to larger voter groups. Under this view, the only representation axiom from this literature that is conceptually compatible with diversity is the most basic one, Justified Representation (JR). JR is a rather weak axiom, which has been strengthened significantly with the goal of capturing proportionality. In this paper, we explore whether it can also be strengthened towards guaranteeing more diversity.

\paragraph{Diverse representation: state of the art} 
The only established formalization of diverse representation in committee elections is based on the Chamberlin-Courant score \cite{chamberlin_representative_1983}. 
In \emph{approval-based} committee elections, this amounts to maximizing the number of voters with at least one representative in the committee. This conception of diverse representation has been thoroughly studied also from the computational point of view.
Indeed, the Chamberlin-Courant problem is equivalent to the unit-weight Maximum k-Coverage problem. \textsc{Max $k$-Coverage} is NP-hard and a greedy algorithm attains a multiplicative $(1-1/\e)$-approximation ratio \cite{hochbaum_analysis_1998}, which is optimal for a polynomial time algorithm \cite{feige_threshold_1998}. Importantly, (under Gap ETH) the problem remains inapproximable even for small $k$ \cite{manurangsi_tight_2019}.

Proportional representation in approval-based multi-winner voting has received a lot of attention over the past years. One focus has been the introduction of a multitude of proportionality axioms, mostly based on the concept of justified representation \cite{aziz_justified_2017} or priceability \cite{PetersS20} and rules satisfying them \cite{peters_proportional_2020}. We refer the reader to the book by \citet{lackner_multi-winner_2023} for a complete overview. 

Strong and Semi-Strong (E)JR were introduced in \citet{aziz_justified_2017}. The complexity of Strong EJR was recently studied by \citet{ai_strong_2026}, who proved $\Theta_2^p$-completeness.

\subsubsection*{Our contribution} 
This paper develops a comprehensive theory of diverse representation. It does so by making three main contributions. Our first contribution is conceptual, in that we view diverse representation as representing many \emph{groups} of voters, a view that had previously only been taken in the proportionality literature.

Second, we study optimization variants of the two known axioms Strong JR (SJR) and Semi-Strong JR (SSJR) introduced by \citet{aziz_justified_2017},
and complete the study of the computational properties of those axioms. It is well-known that JR is efficiently satisfiable and that SJR and SSJR are not always satisfiable \cite{aziz_justified_2017}. Yet, it was not known whether SJR and SSJR could be optimized efficiently. We settle this open question: a greedy algorithm is $(1-1/\e)$-approximate for SJR, which is best possible unless P = NP, whereas for SSJR no polynomial-time algorithm is $\alpha$-approximate for any constant $\alpha>0$ under the Small Set Expansion Hypothesis, and, for some constant $c$, not even $m^{-1/(\log\log m)^c}$-approximate under ETH. We do prove efficient optimization \emph{is} possible on the Candidate Interval and Voter Interval domain. Moreover, we systematically study the logical relations between CC, SJR, SSJR, their optimization variants, JR, and a new optimization problem we call CNASH, based on the Nash-welfare of cohesive groups. The resulting logical relationships are summarized in Figure \ref{figure:relations}.  

Third, in an attempt to capture the unique and defining opinions of a group, we introduce a novel axiom for diverse representation that we call Distinctive Representation (DR). DR is a strengthening of JR, where we require candidates to uniquely represent cohesive groups. We prove DR-optimization in polynomial time to be infeasible (unless P = NP). Since the locally optimal variant \emph{Local DR} \emph{can} be achieved in polynomial time, we use the latter axiom to show how this new diversity notion differs from known diversity and proportionality axioms. 

Our experiments in \Cref{sec:exp} show that, although insatisfiable in theory, SSJR is often satisfiable in practice, and 
committees with strong (S)SJR scores can be attained using polynomial-time algorithms. 
Moreover, we see that Local DR performs particularly well regarding multiple empirical measures of diversity.

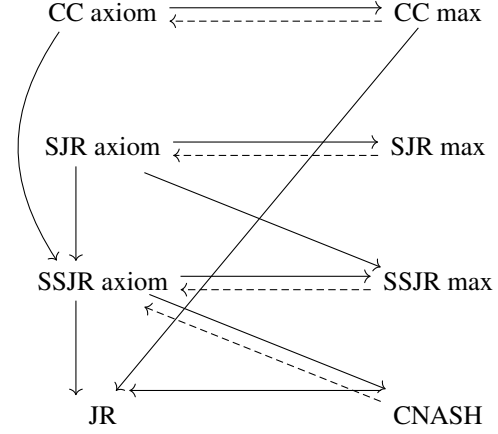
\begin{figure}[t]
\adjustbox{scale=1,left}{%
$$\begin{tikzcd}
	{\text{CC axiom}} &&& {\text{CC max}} \\
    \\
	{\text{SJR axiom}} &&& {\text{SJR max}} \\
	\\
	{\text{SSJR axiom}} &&& {\text{SSJR max}} \\
    \\
    {\text{JR axiom}} &&& {\text{CNASH}}
	\arrow[shift left, draw={black}, from=1-1, to=1-4]
	\arrow[shift left, draw={black}, from=1-4, to=1-1, dashed]
	\arrow[shift right=5, color={black}, curve={height=20pt}, from=1-1, to=5-1]
	\arrow[shift left, draw={black}, from=3-1, to=3-4]
    \arrow[shift left, draw={black}, from=3-4, to=3-1, dashed]
	\arrow[shift right=4, draw={black}, from=3-1, to=5-1]
	\arrow[shift right, draw={black}, from=3-1, to=5-4]
	\arrow[shift left, draw={black}, from=5-1, to=5-4]
    \arrow[shift left, draw={black}, from=5-4, to=5-1, dashed]
    \arrow[shift right=4, draw={black}, from=5-1, to=7-1]
    \arrow[shift right, draw={black}, from=1-4, to=7-1]
    \arrow[shift right=4, draw={black}, from=7-4, to=7-1]
    \arrow[shift left, draw={black}, from=7-4, to=5-1, dashed]
    \arrow[shift left, draw={black}, from=5-1, to=7-4]
\end{tikzcd}$$}
\caption{Relations between axioms and optimizations. Normal arrows denote logical implications. A dashed arrow denotes `implication whenever satisfiable'. Logical implications may follow from transitivity, but when there is no path from one node to another, this denotes logical independence.
}
\label{figure:relations}
\end{figure}

\section{Preliminaries}\label{sec:prelim}
An instance $I$ of an approval-based committee election problem (from hereon: instance) consists of a set $N$ of $n$ voters $v_i$, a set $C$ of $m$ candidates $c_i$, and a natural number $k$. Moreover, every voter approves of a subset of the candidates. This approval information is expressed in the form of an \emph{approval set} $A(i)$ for each voter $v_i\in N$. The sequence of sets $A=(A(1),\ldots,A(n))$ is called the \emph{approval profile}. For any candidate $c\in C$ we denote $N_c:=\{v_i:c\in A(i)\}$. The goal is to elect a committee $W\subseteq C$ of size $k\le m$, based on the approval information submitted by the voters in $N$.

The well known Chamberlin-Courant score is often used as a measure of diversity:
\begin{definition}[Chamberlin and Courant 1983]\label{def:cc}
On any approval-based committee election instance, the \emph{Chamberlin-Courant (CC)} score of a committee $W\subseteq C,\,|W|=k$ is defined as
\begin{align}
\text{CC}(W):=\frac{1}{n}\sum_{i=1}^n\mathds{1}_{\{A(i)\cap W\neq\emptyset\}}(i).
\end{align}
In the Chamberlin-Courant decision problem, the input is an instance $I$ and an integer $x$, and the question is whether it is possible to achieve score $\frac{x}{n}$ using $k$ candidates. In the Chamberlin-Courant optimization problem, the input is an instance $I$ and the task is to maximize the score by electing $k$ candidates.
\end{definition}

For a maximization problem with a non-negative objective function $s$  (that maps an algorithm's output to an objective function value -- a score), we write $\opt_I=\opt$ for the optimal objective value on a given instance $I$. For a function $\alpha$ that maps an instance $I$ to a value
$\alpha(I)\in(0,1]$, we call an algorithm \emph{$\alpha$-approximate} if on every
instance $I$ it outputs a solution $W$ with $s(W)\ge\alpha(I)\cdot\opt_I$.

Since the Chamberlin-Courant decision problem is NP-complete, the prevailing approach for applications is the greedy approximation algorithm with approximation ratio $1-1/\e$, which is also optimal (unless P = NP). From now on, by the \textit{CC axiom} we refer to the requirement to select a committee (of the given size $k$) which achieves \emph{full coverage}. When we speak of the problem of obtaining \emph{optimal CC score}, we mean finding, for a given instance, a $k$-sized committee which achieves the highest possible CC score. 

The proportionality literature has long defined representation according to groups of voters using the following notion of cohesiveness.
\begin{definition}[\textbf{$\ell$-Cohesive groups}, \cite{azizComplexityExtendedProportional2018}]
For $\ell\ge 1$, a group $N'\subseteq N$ is $\ell$-cohesive if $|N'|\ge\ell\frac{n}{k}$ and $|\cap_{v_i\in N'}A(i)|\ge\ell$.
\end{definition}

In the context of diversity, we only consider $\ell=1$ and speak of a \emph{cohesive group} (instead of 1-cohesive). Since such a group can only be formed around a candidate, we will speak of such a candidate as a \emph{cohesive candidate}: candidate $c$ is cohesive when $|N_c|\ge n/k$. We can now define the first axiom using cohesive groups.
\begin{definition}[\textbf{Justified Representation}, \cite{aziz_justified_2017}]
A committee $W$ satisfies Justified Representation (JR) if for all cohesive groups $N'\subseteq N$, there exists $v_i\in N'$ such that $|A(i)\cap W|\ge 1$.	
\end{definition}

Justified Representation is satisfiable in polynomial time, meaning that in every approval-based committee election instance, we can efficiently find a committee abiding by the axiom. Indeed, greedily maximizing the CC score generates a committee satisfying JR \cite{aziz_justified_2017}.
However, JR is a very weak axiom. Consider the instance shown in Figure~\ref{fig:intro}~\subref{fig:intro-a}. It has $n=24$ voters, so for $k=4$, any candidate with six approvals is cohesive. By electing $W=\{a,b,c,d\}$, full voter coverage is attained. Moreover, this committee represents \emph{all} voters of each cohesive group, \emph{and} it represents them by a candidate that binds them. However, any committee containing (the non-cohesive candidate) $e$ satisfies JR. Precisely these intuitions bring us to two known strengthenings of JR. 

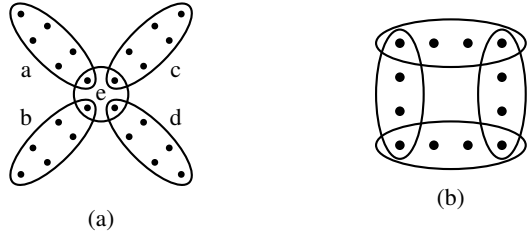
\begin{figure}[t]
\centering
\begin{subfigure}[]{0.45\columnwidth}\centering
    \begin{tikzpicture}[scale=0.3, thick]
    \draw (0,0) circle (1.2);
    \foreach \angle in {45, 135, 225, 315} {
        \begin{scope}[rotate=\angle]
            \draw (3.0, 0) ellipse (2.5 and 0.9);
            \fill (0.85, 0) circle (0.15);
            \fill (2.2, 0.45) circle (0.15);
            \fill (2.2, -0.45) circle (0.15);
            \fill (3.8, 0.45) circle (0.15);
            \fill (3.8, -0.45) circle (0.15);
            \fill (5.0, 0) circle (0.15);
        \end{scope}
    };
        \node[font=\small] at (-3.3, 1) {\text{a}};
        \node[font=\small] at (-3.3, -1) {\text{b}};
        \node[font=\small] at (3.3, 1) {\text{c}};
        \node[font=\small] at (3.3, -1) {\text{d}};
        \node[font=\small] at (0, 0) {\text{e}};
    \end{tikzpicture}
    \caption{}
    \label{fig:intro-a}
    \end{subfigure}
\hfill
\begin{subfigure}[]{0.45\columnwidth}\centering
    \begin{tikzpicture}[scale=0.45, thick]
        \draw (1.5, 3) ellipse (2.2 and 0.7);
        \draw (1.5, 0) ellipse (2.2 and 0.7);
        \draw (0, 1.5) ellipse (0.7 and 1.9);
        \draw (3, 1.5) ellipse (0.7 and 1.9);
        \foreach \x in {0, 1, 2, 3} {
            \fill (\x, 3) circle (4pt);
            \fill (\x, 0) circle (4pt);
        }
        \foreach \y in {1, 2} {
            \fill (0, \y) circle (4pt);
            \fill (3, \y) circle (4pt);
    }
\end{tikzpicture}
    \caption{}
    \label{fig:intro-b}
    \end{subfigure}
\caption{This figure shows two approval instances, where voters are depicted as nodes and candidates as sets. A voter approves of a candidate iff they are contained in the set.}
\label{fig:intro}
\end{figure}

\begin{definition}[\textbf{Strong Justified Representation}, \cite{aziz_justified_2017}]
A committee $W$ satisfies Strong Justified Representation (SJR) if for all cohesive groups $N'\subseteq N$ we have $W\cap(\cap_{v_i\in N'}A(i)) \neq \emptyset$.
\end{definition}

Less strong is requiring that all voters in a cohesive group are represented, but not necessarily by the same candidate, which gives the following axiom.

\begin{definition}[\textbf{Semi-Strong Justified Representation} \cite{aziz_justified_2017}]
A committee $W$ satisfies Semi-Strong Justified Representation (SSJR) if for all cohesive groups $N'\subseteq N$ we have $W\cap A(i) \neq \emptyset$ for all $v_i\in N'$.
\end{definition}

\section{(S)SJR, CC and CNASH}\label{sec:ssjr-cc}
In this section, we take the SJR and SSJR axioms from \Cref{sec:prelim}, and a new rule that we call CNASH, and study how they logically relate to the CC axiom. We do the same for optimization variants of the axioms. Indeed, it is not always possible to satisfy either (S)SJR or CC. Consider the known counterexample for (S)SJR in Figure~\ref{fig:intro}~\subref{fig:intro-b}. It has $n=12$ voters, so for $k=3$ we have $n/k=4$. Then any group of four voters that approves a shared candidate, is cohesive, implying that all four candidates require representation, but we can only elect $k=3$. This motivates us to define maximization variants of both SJR and SSJR. 

\begin{definition}[\textbf{Strong JR maximization}]\label{def:sjr-max}
On any approval-based committee election instance, the SJR-score of a committee $W\subseteq C, |W|=k$ is
$$\SJR(W):=\frac{|\{c\in C: |N_c|\ge\frac{n}{k} \land \cap_{v_i\in N_c}A(i) \cap W \neq \emptyset\}|}{|\{c\in C:|N_c|\geq \frac nk\}|} \ ,$$
when $|\{c\in C:|N_c|\geq \frac nk\}|> 0$, otherwise define $\text{SJR}(W)=1$ for any $W$. The SJR maximization problem is defined as finding $W\subseteq C,|W|=k$ maximizing $\SJR(W)$.
\end{definition}
Note that alternatively, one can write the numerator in Definition~\ref{def:sjr-max} as $|\{c\in C: |N_c|\ge\frac{n}{k} \land \exists d\in W: N_c\subseteq N_d\}|$.

The name comes from the following proposition.
\begin{proposition}\label{prop:sjr-score}
We have $SJR(W)=1 \iff W$ satisfies Strong JR.	
\end{proposition}

\begin{definition}[\textbf{Semi-Strong JR maximization}]\label{def:ssjr-max}
On any approval-based committee election instance, the SSJR-score of a committee $W\subseteq C, |W|=k$ is
$$\SSJR(W):=\frac{|\{c\in C:|N_c|\ge\frac{n}{k},\,\forall v_i\in N_c: A(i)\cap W\neq\emptyset \}|}{|\{c\in C:|N_c|\geq \frac nk\}|} \ ,$$
when $|\{c\in C:|N_c|\geq \frac nk\}|> 0$, otherwise define $\text{SSJR}(W)=1$ for any $W$. We define the SSJR maximization problem as finding $W\subseteq C,|W|=k$ maximizing $\SSJR(W)$.
\end{definition}
Note that alternatively, we can write the numerator in Definition~\ref{def:ssjr-max} as $|\{c\in C: |N_c|\ge\frac{n}{k} \land  N_c \subseteq \cup_{c'\in W} N_{c'}\}|$.

Again, the name is justified by the following proposition.
\begin{proposition}\label{prop:ssjr-score}
We have $\SSJR(W)=1 \iff W$ satisfies Semi-Strong JR.	
\end{proposition}

Of course, we can also try to maximize other metrics instead of pure coverage. In many domains,
maximizing Nash-welfare leads to particularly good outcomes \cite{CaragiannisKMPS19}. Nash-welfare in committee 
elections has so far only been considered on the level of individual voters, which leads to the same drawbacks
as considering coverage on an individual level: if every outcome has at least one voter with zero representatives 
in the committee, then every committee offers the same Nash-welfare. On the level of cohesive groups, we can avoid
this problem, as a committee that at least partially covers every cohesive group, i.e., a JR committee, always exists.
This motivates the following definition:

\begin{definition}[Cohesive Nash Rule]
On any approval-based committee election instance, let Coh be the set of
all cohesive groups. Then, the cohesive Nash score of a committee $W\subseteq C, |W|=k$ is
$$\text{NASH}(W):=\prod_{N' \in \text{Coh}} |\{i \in N' \mid A(i) \cap W \neq \emptyset\}| .$$
The Cohesive Nash Rule (CNASH) is defined as the rule that picks a committee $W$ with maximal Nash score.
\end{definition}

We can now study the relations between the different axioms, their optimization variants, and the Cohesive Nash Rule. These relations are summarized in \Cref{figure:relations}, and relevant relations are stated in \Cref{thm:relations}. See the overview in Appendix~\ref{app:relations} for explanations of the relations that do not immediately follow from \Cref{thm:relations}. \Cref{figure:relations} highlights the absence of logical relations between the optimization problems, where the axioms are somewhat hierarchically structured. Interestingly, optimizing for SJR or SSJR does not even yield a committee that satisfies JR, and the CC axiom (full coverage) does not imply SJR. Another notable observation is that both the SSJR maximization and CNASH are functions of, and increase monotonically with, $\cup_{c\in W}N_c$, which makes them coincide when the SSJR axiom is feasible, but unrelated otherwise.

\begin{theorem}\label{thm:relations}
The following relations hold:
\begin{enumerate}
    \item\label{prop:cc-to-ssjr} A committee $W$ that satisfies the CC axiom, also satisfies the Semi-Strong JR axiom.
    \item\label{prop:cc-sjr} A committee $W$ that satisfies the CC axiom, may not be optimal for Strong JR.
    \item\label{prop:sjr_ssjr_cc} A committee $W$ that satisfies the (Semi-)Strong JR axiom, may not be optimal for CC, and vice versa.
    \item \label{prop:cc-ssjr} A committee $W$ that is optimal for CC, may not be optimal for Semi-Strong JR.
    \item\label{sjr-ssjr} A committee $W$ that is optimal for  Strong JR, may not be optimal for Semi-Strong JR.
    \item\label{prop:ssjr-jr} A committee $W$ that is optimal for (Semi)-Strong JR, may not satisfy JR.
    \item\label{prop:cnash-jr} A committee $W$ that is optimal for CNASH, satisfies JR.
    \item\label{prop:cnash-sjr} A committee $W$ that is optimal for CNASH, may fail to find a committee that satisfies the SJR axiom, even when it exists.
    \item\label{prop:cnash-cc} A committee $W$ that is optimal for CNASH, may not be optimal for CC, and vice versa.
    \item\label{prop:jr-to} A committee $W$ that satisfies JR, may not be optimal for SSJR max, SJR max, CC max or CNASH
    \item When satisfying the SSJR axiom is feasible, the set of NASH‑maximizing committees equals the set of SSJR satisfying committees.
    \item A committee $W$ that is optimal for CNASH may not be optimal for SSJR max, and vice versa.
    \item A committee $W$ that is optimal for SJR, may not be optimal for CNASH.
\end{enumerate}
\end{theorem}


\section{Hardness and optimization}\label{sec:ssjr_max}
In this section we study the SJR, SSJR and CNASH decision/optimization problems. For CNASH, we prove the decision problem to be NP-hard. For SJR and SSJR, while we know that we cannot always satisfy the respective axioms, we want to understand how well we can do on the maximization variants. We prove the decision versions of both problems are NP-complete, which implies that we cannot find optimal solutions, unless P = NP. We then turn to find approximate solutions. We find that no polynomial time algorithm can find a solution closer than $1-1/\e$ to the optimal solution (unless P = NP), and give an algorithm witnessing this optimal approximation ratio for SJR. For SSJR, we prove that we cannot expect any constant ratio in polynomial time (under SSEH) and that
(under ETH) there exists a constant $c$ such that no polynomial-time algorithm is $m^{-1/(\log\log m)^c}$-approximate (and similarly with $n$ in the place of $m$). We start with a simple but nevertheless useful observation.

\begin{proposition}\label{prop:SJR-score} For any approval-based committee election instance and committee $W$, the Semi-Strong JR maximization score and the Strong JR maximization score of $W$ can both be computed in polynomial time.	
\end{proposition}

Below we define the decision problems that we will prove to be NP-hard. For a definition of the problems that we reduce from, we refer to the Appendix.

\begin{definition}[Semi-Strong JR decision problem]
In the \textsc{Semi-Strong JR decision problem (SSJR)} we are given an instance $I=(N,C,A,k,t)$, where $N$ is a group of $n$ voters, $C$ is a set of $m$ candidates, $A:N\to \mathcal{P}(C)$ is a function that maps every voter to their approval set $A(i)$ of candidates, $k\in\{1,\ldots,m\}$ is the committee size and $t\in \mathbb{N}$. Then $I\in \textsc{SSJR}\iff \exists W\subseteq C \text{ s.t. } |W|=k \text{ and } |\{c\in C:|N_c|\ge\frac{n}{k},\,\forall v_i\in N_c: A(i)\cap W\neq\emptyset \}|\ge t$.
\end{definition}

\begin{definition}[Strong JR decision problem]
In the \textsc{Strong JR decision problem (SJR)} we are given an instance $I=(N,C,A,k,t)$, where $N$ is a group of $n$ voters, $C$ is a set of $m$ candidates, $A:N\to \mathcal{P}(C)$ is a function that maps every voter to their approval set $A(i)$ of candidates, $k\in\{1,\ldots,m\}$ is the committee size and $t\in \mathbb{N}$. Then $I\in \textsc{SJR}\iff \exists W\subseteq C \text{ s.t. } |W|=k \text{ and } |\{c\in C: |N_c|\ge\frac{n}{k} \land \cap_{v_i\in N_c}A(i) \cap W \neq \emptyset\}|\ge t$.
\end{definition}

\begin{definition}[Cohesive Nash decision problem]
In the \textsc{Cohesive Nash decision problem (CNASH)} we are given an instance $I=(N,C,A,k,t)$, where $N$ is a group of $n$ voters, $C$ is a set of $m$ candidates, $A:N\to \mathcal{P}(C)$ is a function that maps every voter to their approval set $A(i)$ of candidates, $k\in\{1,\ldots,m\}$ is the committee size and $t\in \mathbb{N}$. Then $I\in \textsc{CNASH}\iff \exists W\subseteq C \text{ s.t. } |W|=k \text{ and } \prod_{N' \in \text{Coh}} |\{i \in N' \mid A(i) \cap W \neq \emptyset\}| \ge t$ where Coh is the set of all cohesive groups.
\end{definition}
\begin{theorem}\label{thm:SSJR_hard}
The SSJR decision problem is NP-complete.
\end{theorem}
\begin{theorem}\label{thm:sjr_hardness}
The SJR decision problem is NP-complete.
\end{theorem}
\begin{theorem}\label{thm:cnash_hardness}
The CNASH decision problem is NP-hard.
\end{theorem}


If we cannot find an optimal solution, we can still hope to approximate the optimal solution well. The following result limits the best approximation ratio we may hope to obtain for these two problems.

\begin{theorem}\label{thm:SJR_opt_approx}
Electing $k$ candidates, no polynomial time algorithm can approximate the optimal solution for (S)SJR with ratio better than $1-1/\e$, unless P = NP.
\end{theorem}

We now give an algorithm which achieves this optimal approximation ratio for SJR.

\begin{algorithm}[h]
	\SetAlgoNoLine
	\KwIn{Numbers $n,m,k\in\N$ with $k\le m$, set $N$ of $n$ voters, set $C$ of $m$ candidates, approval profile $A$}
	\KwOut{Committee $W\subseteq C$ of size $k$.}
    Let $W=\{\}$ be an empty set\;
    	\For{$i=1,\ldots,k$}{
            \For{$c\in C \backslash W$}{
                Compute the increase in (S)SJR-score if $c$ were added to $W$. Concretely, iterate over all candidates $c'\in C$ with $|N_{c'}| \ge \frac nk$ and check:

                For SSJR: if $N_{c'} \subseteq \cup_{d\in W \cup \{c\}}N_d$ while not $N_{c'} \subseteq \cup_{d\in W}N_d$

                For SJR: if $\cap_{v_i\in N_{c'}}A(i) \cap (W\cup \{c\}) \neq \emptyset$ while $\cap_{v_i\in N_{c'}}A(i) \cap W = \emptyset$.

                The increase in score equals the number of candidates $c'$ for which the above conditions holds.

        }
        Pick $c$ that maximizes the increase in score, and add it to $W$: $W=W\cup \{c\}$.
    }
    \caption{\textsc{greedy-(s)sjr}}
    \label{alg:greedy2-ssjr}
\end{algorithm}


\begin{theorem}\label{thm:ratio}
Algorithm~\ref{alg:greedy2-ssjr} is $(1-1/\e)$-approximate for SJR and runs in polynomial time. 
\end{theorem}

Interestingly, it is not hard to see that the greedy CC algorithm  fails to come close to this optimal approximation ratio of SJR. To see this, note that the greedy CC algorithm always picks the candidate increasing the CC score most in that round. Therefore in the first round it will always pick the largest candidate (which will be cohesive, if a cohesive group exists at all), which assures a score of at least $\nicefrac{1}{|\{c : |N_c|\geq\nicefrac{n}{k}\}|}$, for committee size $k$. However, one can construct an instance where this first choice blocks all other cohesive groups from being selected by greedy CC. For any committee size $k$, the ratio will therefore not exceed $1/k$. 

Although SSJR is a less strict requirement than SJR, its approximation is, under the extra assumptions of Small Set Expansion Hypothesis and Exponential Time Hypothesis, significantly more difficult, as we show below. 
\begin{theorem}\label{thm:ssjr-approxratios}
(i) Under the Small Set Expansion Hypothesis (SSEH), there exists no constant $\alpha\in(0,1]$ such that a polynomial-time $\alpha$-approximation algorithm exists for SSJR maximization. 
(ii) Under the Exponential Time Hypothesis (ETH), there exist constants $c,d>0$ such that no polynomial-time algorithm is $m^{-1/(\log\log m)^c}$-approximate for SSJR maximization, and no polynomial-time algorithm is $n^{-1/(\log\log n)^d}$-approximate for SSJR maximization.
\end{theorem}

\section{(S)SJR on restricted domains}\label{sec:domains}
Motivated by the insatisfiability of SJR and SSJR on general domains, we study satisfiability on restricted domains. Notably, most satisfiability results are negative, which stands in contrast with properties of the core \cite{PierczynskiS22core}. We consider four types of domain restrictions.

In a Candidate Interval (CI) domain, there exists an ordering of the candidates so that for each voter $v_i$, the candidates in $A(i)$ form an interval. In a Candidate Extremal Interval (CEI) domain, an extra requirement that all candidates from $C\setminus A(i)$ also form an interval, implies each voter's interval `touches' at least one end of the linear order of candidates. Conversely, in a Voter Interval (VI) domain, there exists an ordering of the voters so that each candidate is approved by an interval of voters, and the Voter Extremal Interval (VEI) domain is defined in an analogous way. Formal definitions can be found in Appendix~\ref{app:domains}. A CI domain occurs in scenarios where the \emph{candidate} space is one-dimensional, and so candidates can be ordered linearly, e.g. by ideology or location. A VI domain fits scenarios where \emph{voters} are structured to form a linear ordering, e.g. on the basis of demographic or socioeconomic features. 
\citet{brill2025individual} showed that SSJR is always satisfiable on the CEI and VEI domains but that the same does not hold for CI and VI domains. The latter fact also means that SJR cannot always be satisfied on the CI and VI domains. We complete the picture for SJR with the following result.


\begin{theorem}\label{thm:sjr-domains}
SJR is always satisfiable on the VEI domain, but not on the CEI domain.
\end{theorem}

Although satisfiability remains problematic, we see more positive results for efficient maximization.

\begin{theorem}\label{thm:(S)SJR-max-CI-VI-poly}
Maximizing the (S)SJR-score can be efficiently done on the CI and VI domain.
\end{theorem}

However, in many cases there will be some problematic candidates that keep the domain from being Candidate Interval. Luckily, collecting them in a set, we can find an FPT algorithm for SJR optimization anyway. Given an election instance $I = (N,C,A,k)$, we say that $D\subseteq C$ is a \emph{candidate-deletion set} if the removal of these candidates makes the election $I' = (N,C\setminus D,A,k)$ have CI preferences.

\begin{theorem}\label{thm:fpt}
Assuming we are given the candidate-deletion set $D$, there exists an FPT algorithm to perform SJR maximization with respect to $|D|$.
\end{theorem}

\section{Distinctive Representation}\label{sec:DR}
While SJR and SSJR guarantee that all cohesive groups are represented by \emph{some} candidate, they do not provide any guarantees on the quality of this representation. In particular, representing candidates do not necessarily have to capture the unique and defining opinions of the group. Instead, we might end up in a situation where all groups are represented by benign candidates that are broadly acceptable to everyone, as selecting just one uniformly approved candidate suffices to satisfy SJR and SSJR. To circumvent this problem, we introduce a new diversity axiom that strengthens JR in a different direction, and that we call \emph{Distinctive Representation} (DR).



            




\begin{definition}[Distinctive Representation]
    A committee $W$ satisfies \emph{Distinctive Representation} (DR) if $(i)$ it satisfies JR, and $(ii)$ $|N_{c}| = \lceil\nicefrac{n}{k}\rceil$ holds for every candidate $c\in W$.
\end{definition}


\begin{observation}\label{prop:dr-sat}
Not all instances allow for a committee that satisfies Distinctive Representation. Even on CEI/VEI instances, DR is insatisfiable.
\end{observation}

This insatisfiability makes optimization a natural objective. However, we will see next that (unless P = NP), for no non-negative measure $M$ that equals zero exactly on DR committees, a polynomial-time algorithm that always outputs a committee minimizing $M$, exists.

\begin{definition}[DR Decision Problem]
In the \textsc{Distinctive Representation} decision problem we are given an instance $I=(N,C,A,k)$, where $N$ is a group of $n$ voters, $C$ is a set of $m$ candidates, $A:N\to \mathcal{P}(C)$ is a function that maps every voter to their approval set $A(i)$ of candidates, and $k\in\{1,\ldots,m\}$ is the committee size. Then $I\in \textsc{DR}\iff \exists W\subseteq C \text{ s.t. } |W|=k\, \land\, W \text{ satisfies JR }\land\, |N_c|=\lceil n/k\rceil \forall c\in W$.
\end{definition}

\begin{theorem}\label{thm:dr-hard}
The Distinctive Representation Decision Problem is NP-complete.
\end{theorem}

\section{Local DR}\label{sec:LocalDR}
\Cref{thm:dr-hard} says we cannot optimize DR in polynomial time (unless P = NP). Instead of studying approximations for DR as an optimization problem\footnote{Indeed, it is not obvious with respect to which function this would be. We could, for instance, take a sum $\sum_{c\in W}||N_c|-\lceil n/k\rceil|$ or a maximum $\max_{c\in W}||N_c|-\lceil n/k\rceil|$.}, we define a slightly weaker, locally optimal axiom that \emph{is} satisfiable and that we can satisfy in polynomial time: Local Distinctive Representation.
\begin{definition}[Local DR]
A committee $W$ satisfies \emph{Local Distinctive Representation} (local DR) if 
    \begin{enumerate}
        \item $W$ satisfies JR, and
        \item there exist no two candidates $c\in W,d\notin W$ such that $||N_{c}|-\lceil n/k\rceil|> ||N_{d}|- \lceil n/k\rceil|$ while $(W\cup\{d\})\setminus\{c\}$ still satisfies JR, unless $|N_c|\ge n/k$ and $|N_d|<n/k$.
    \end{enumerate}  
\end{definition}

We present below a local search algorithm that satisfies Local DR. For any two candidates $c\in W,\, c'\notin W$, write $\Delta(W,c',c):=||N_c|-\lceil n/k\rceil|-||N_{c'}|-\lceil n/k\rceil|$.

\begin{algorithm}[h]
	\SetAlgoNoLine
	\KwIn{Numbers $n,m,k\in\N$ with $k\le m$, set $N$ of $n$ voters, set $C$ of $m$ candidates, approval profile $A$.}
	\KwOut{Committee $W\subseteq C$ of size $k$.}
    Choose any random starting committee $W\subseteq C$ with $|W|=k$ that satisfies JR (e.g. using greedy-CC rule)
    \;
    	\Repeat{$\Delta(W,c',c)\le 0$}{
        \begin{enumerate}
            \item From all $(c',c)$ for which $(W\cup \{c'\})\setminus\{c\}$ satisfies JR and for which $|N_{c'}|<n/k\implies |N_c|<n/k$, pick $(c',c)$ that maximizes $\Delta(W,c',c)$. If no feasible pair $(c',c)$ with $\Delta(W,c',c)>0$ exists, terminate.
            \item $W=(W\cup \{c'\})\setminus\{c\}$ if $\Delta(W,c',c)>0$.
        \end{enumerate}
        }
    \caption{\textsc{ls-DR}}
    \label{alg:ls-DR}
\end{algorithm}

\begin{proposition}\label{prop:local-dr-alg}
The LS-DR algorithm satisfies Local DR and runs in polynomial time.
\end{proposition}

\begin{proposition}\label{prop:find-dr-committee}
There exist instances where a DR committee exists, but LS-DR may fail to find it.
\end{proposition}

Since DR is designed as a diversity axiom, it is not surprising that it does not behave the same as known proportionality axioms such as EJR \cite{aziz_justified_2017} and PJR \cite{sanchez-fernandez_proportional_2017}.
Although Local DR and EJR/PJR can coexist (for an instance with no $\ell$-cohesive groups for $\ell>1$, EJR/PJR is equal to JR, and DR also satisfies JR), neither EJR nor PJR implies Local DR, nor vice versa. 
We also consider its relationship to priceability \cite{PetersS20} and (S)SJR.
\begin{theorem}\label{thm:DR-proportionality}
Local DR does not imply EJR or PJR and vice versa. 
 DR does not imply priceability, and priceability does not imply Local DR.
Local DR does not imply (S)SJR (even when satisfiable), and (S)SJR does not imply Local DR.
\end{theorem}

A rule that seems similar to (Local) DR is Monroe's rule \cite{monroe1995fully}, where a function $\phi:N\to W$ maps voters to committee members in such a way that $\lfloor n/k\rfloor\le |\phi^{-1}(c)|\le \lceil n/k\rceil$, and voters only gain utility when $\phi(v_i)\in A(i)$. There lies the difference with Local DR, where we care about the \emph{size} of the candidates being (close to) $\lceil n/k \rceil$.

\begin{proposition}\label{prop:monroe-dr}
An optimal Monroe committee need not satisfy (Local) DR, and vice versa.
\end{proposition}

Another axiom from the literature that seems similar to DR, is Balanced JR, introduced by  \citet{fish2024generative}. The axiom is introduced for cardinal utilities, but for approval utilities and assuming the available candidates form a set (instead of a multi-set, which is possible in the generative social choice context) it boils down to the following definition.

\begin{definition}[Balanced JR \cite{fish2024generative} for approval ballots]
A committee $W$ satisfies Balanced Justified Representation (BJR) for approval ballots if there is a function $\omega:N\to W$, matching voters to candidates, such that each candidate on the committee is matched to either $\lfloor n/k\rfloor$ or $\lceil n/k\rceil$ voters, for which there is no coalition $S\subseteq N$ and candidate $c \notin W$ such that both (i) $|S|\ge n/k$, (ii) $c\in A(i)$ for all $i\in S$ and (iii) $\omega(i)\notin A(i)$ for all $i\in S$.
\end{definition}

We can view BJR for approval ballots as JR with an assignment function, making it a strengthening of JR that, similar to DR, aims to avoid the problem that \emph{one} uniformly approved candidate suffices to represent everyone. However, in contrast to DR, voters can be represented by candidates that are not distinctive to them, but instead are broadly supported. DR and BJR thus strengthen JR in different directions; BJR adds a balanced assignment, and DR adds distinctiveness of the elected candidates themselves, as we prove in the following theorem.
\begin{theorem}\label{thm:DR-BJR}
\begin{itemize}
        \item\label{prop:DR-to-BJR}
        A committee $W$ that satisfies DR (when satisfiable) needn't satisfy BJR for approval ballots, and vice versa.
        \item\label{prop:LDR-BJR} 
        A committee $W$ that satisfies Local DR needn't satisfy BJR for approval ballots, and vice versa.
\end{itemize}
\end{theorem}

\section{Experiments}\label{sec:exp}

\begin{figure*}[t]
\centering
\begin{subfigure}[b]{0.315\textwidth}
        \centering
        \includegraphics[width=\linewidth]{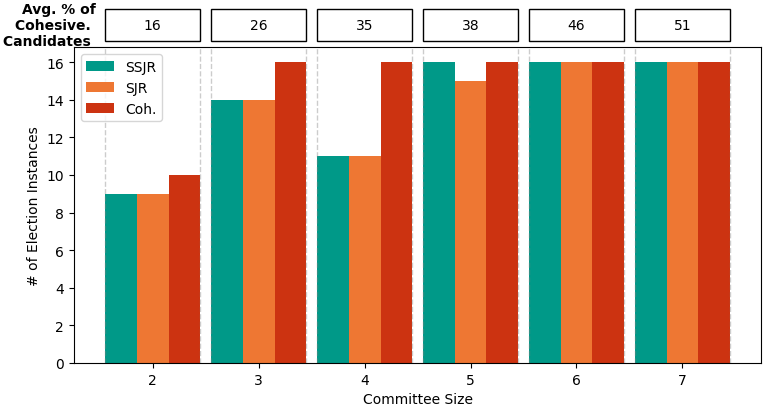}
        \caption{(S)SJR existence: Situ.}
        \label{fig:s_sjr_situ}
    \end{subfigure}~\begin{subfigure}[b]{0.315\textwidth}
        \centering
\includegraphics[width=\linewidth]{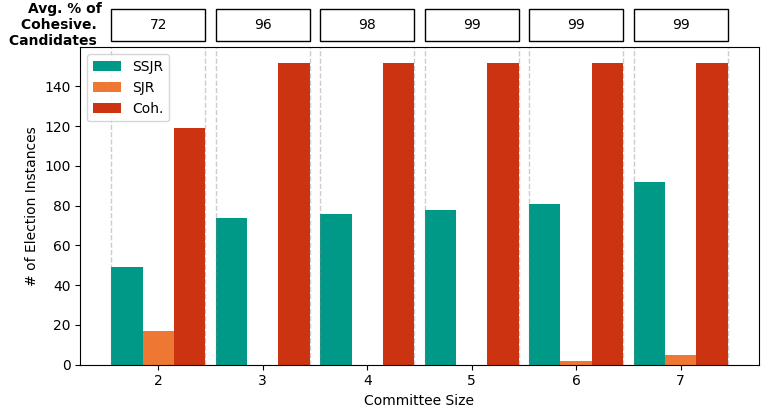}
        \caption{(S)SJR existence: SCOTUS.}
        \label{fig:s_sjr_scotus}
    \end{subfigure}~\begin{subfigure}[b]{0.315\textwidth}
        \centering
        \includegraphics[width=\linewidth]{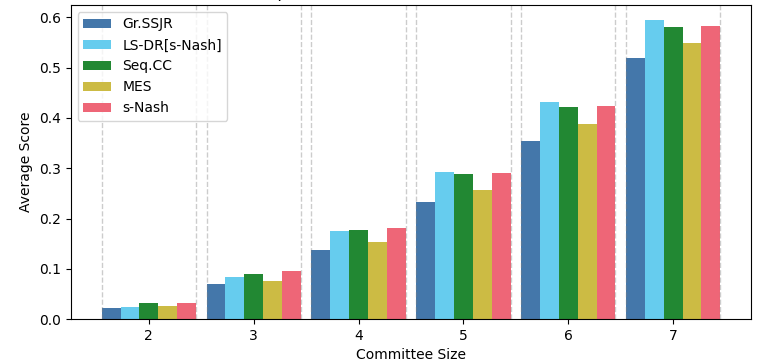}
        \caption{Unique candidate scores: SCOTUS.}
        \label{fig:uniq_cand_scotus}
    \end{subfigure}
\caption{Subfigure $(a)$ shows (S)SJR existence for the Situ dataset. Subfigures $(b)$ and $(c)$ show (S)SJR existence and the unique candidate scores for the SCOTUS data set,  respectively.}
\label{fig:exp_results}
\end{figure*}

Thus far, we have illustrated the difficulty of guaranteeing strong diversity notions. However, an experimental analysis can reveal a more fine-grained and positive outlook than that provided by theory. So, in this section, we go beyond our theoretical results and analyse how our considered notions (both axioms and rules) behave in practice. 

To do so, we tested them against both real-world and synthetic datasets. For real-world data, we considered the \textit{Opinions of the United States Supreme Court (SCOTUS)} \cite{spaeth_2023supreme} and \textit{Voter Autrement in Situ (Situ)} \cite{baujard_2010,baujard_2025_2007,baujardGHLD_2025_2012,baujardGHLL_2013,bouveret_2019} datasets.\footnote{Both datasets were obtained from \emph{PrefLib} \cite{MaWa13a} (\url{https://preflib.github.io/PrefLib-Jekyll/}).} The SCOTUS dataset has $152$ profiles with $n\in \{60,\ldots,277\}$ and $m\in \{8,\ldots,11\}$ while in the $16$ Situ profiles we had $n\in \{215,\ldots,1291\}$ and $m\in \{10,11,12\}$. For each profile in the datasets, we conduct our experiments for committee sizes $k$ ranging between $2-7$. We also assessed synthetic data that was generated using the \emph{Resampling Model} \cite{szufa_how_2022}. The parameters used were $n=100,m=20,p=0.2$ and $\phi\in\{0.25,0.5,0.75\}$ where we generated $5$ election profiles per choice of $\phi$. These specific parameters were chosen to create a variety of more sparse datasets having voters approving a lower fraction of the candidates in contrast to the two real-world datasets.

We analyse the following voting rules: Greedy CC, MES \cite{PetersS20} completed with sequential Phragm\'{e}n, Greedy (S)SJR, a simplified version of the Cohesive Nash rule that only considers maximal cohesive supporter groups of each candidate\footnote{Note that Theorem~\ref{thm:cnash_hardness} indeed only proved hardness of CNASH; membership in NP remains unclear, since $|\text{Coh}|$ can be exponential. When computing the simplified rule, we enumerate all maximal cohesive groups.}, and LS-DR that uses the simplified Cohesive Nash rule to obtain the initial committee\footnote{Note that this simplified Cohesive Nash rule may not return a committee satisfying JR. However, this did not occur with any of our considered election instances.}. 

\subsection{(S)SJR Existence, and Performance of  Algorithm~\ref{alg:greedy2-ssjr}}
Given that the SSJR and SJR axioms are not always satisfiable, we first study how often it is possible to satisfy SSJR and SJR in practice. We ask this for the election instances where there exists some cohesive group. The algorithms to determine this were encoded as ILPs, the details of which are in Appendix~\ref{app:exp}. As we say that a candidate $c$ is cohesive if $N_c$ is a cohesive group, we also checked, for each election instance, what percentage of the candidates are cohesive. Then for each committee size $k$, we averaged these values over the profiles to obtain the average rate of cohesive candidates. We deem this rate as a proxy for the difficulty of achieving diverse group representation via (S)SJR, as intuitively, the more cohesive groups there are, the harder it is to represent them all beyond what is achievable with JR. 

Additionally, we have shown that the greedy (S)SJR algorithm is $(1-\nicefrac{1}{\e})$-approximate for SJR, while for SSJR, the greedy algorithm does not have this guarantee, and the results of Theorem~\ref{thm:ssjr-approxratios} suggest nothing close to $1-1/\e$ is possible at all for SSJR in polynomial time. This begs the question of how well both greedy rules perform in practice, which we study secondly. For each election instance, we assessed the (S)SJR-scores obtained by the greedy rules. 

\paragraph{Results:} In general, we see that SSJR is satisfiable much more often than SJR. Observe that for the SCOTUS data, when $k\geq 3$, almost all candidates are cohesive (see Figure~\ref{fig:s_sjr_scotus}). This contrasts the Situ data (see Figure~\ref{fig:s_sjr_situ}) and the synthetic data (plots in Appendix~\ref{app:exp}), where there is a lower portion of cohesive candidates. Thus, it is more challenging to satisfy (S)SJR on the SCOTUS data, whereas (S)SJR committees occur much more frequently on other datasets.

Regarding the performance of Algorithm~\ref{alg:greedy2-ssjr}, our results show that both of the greedy rules consistently return committees with 
good (S)SJR-scores (see Appendix). 
Notably, this also holds for all the other voting rules that we tested,v although to a lesser extent for LS-DR. We elaborate more on this in Appendix~\ref{app:exp}.

\subsection{Diversity Measures}
We also assess the voting rules against multiple measures of diversity. Besides the CC score (Definition~\ref{def:cc}), we consider the following measures:
\begin{itemize} 
    \item \emph{(S)SJR-score:} For an instance $I$ and a committee $W$, the SSJR-score for a committee $W$ is given in Definition~\ref{def:ssjr-max}, and the SJR-score is given in Definition~\ref{def:sjr-max}. 
    \item \emph{Unique voter-score:} For an instance $I$ and a committee $W$, take $\mathcal{A}(S) = \{A(i)\cap S\mid v_i\in N\}$ to be the set of unique approval ballots over the candidate set $S\subseteq C$. The \emph{unique voter} score of $W$ is: $uv(W) = \nicefrac{|\mathcal{A}(W)|}{|\mathcal{A}(C)|}.$ This measure is meant to capture the portion of the unique voter opinions retained by the voting rule.
    \item \emph{Unique candidate-score:} For an instance $I$ and a committee $W$, take $\mathcal{C}(S)$ to be the set of unique candidate supporter sets $N_c$ for candidates $c\in S\subseteq C$. The \emph{unique candidate} score of $W$ is: $uc(W) = \nicefrac{|\mathcal{C}(W)|}{|\mathcal{C}(C)|}.$ If we identify a candidate by their supporter groups, this measure assesses how well a voting rule does in retaining unique voter coalitions.   
\end{itemize}

The latter two measures represent diversity notions that differ from those that are standard in the committee-election literature. Note that as with the cohesiveness checks from the previous section, we averaged these scores over multiple election instances for each committee target size. We omitted Greedy SJR from the results as it usually performed similar to (and at times, was outperformed by) Greedy SSJR on the diversity measures.

\paragraph{Results:} 
Given the prior observation that the SCOTUS data has more cohesive candidates to satisfy, we focus on this dataset for the remainder of our experimental analysis. Further to this choice, we note that when comparing the different datasets' measure results, the differences in the rules' performances are (generally) more noticeable in the SCOTUS dataset. Thus, we relegate (most of) the results for the Situ and synthetic data to the Appendix.

In terms of the unique voter-score, the two standout rules on the SCOTUS data are LS-DR and simplified CNASH (see Figure~\ref{fig:uniq_vote_scotus}). LS-DR also performed well regarding this measure on the Situ and synthetic data (for $\phi\in\{0.25,0.5\}$) along with Greedy SSJR and Greedy CC. For the unique candidate-score on the SCOTUS data, we found that LS-DR, Greedy CC and simplified CNASH are the best performing rules, all being at a similar level (see Figure~\ref{fig:uniq_cand_scotus}).

We also note that despite being geared towards a different type of diversity, LS-DR still performs reasonably in terms of the CC and SJR coverage scores when compared to the other rules. Specifically, we find that adding the DR-focused local-swap dynamics to the simplified CNASH committee does not sacrifice a significant amount in terms of standard diversity notions (see Appendix~\ref{app:exp}).

\section{Conclusions}
This paper is the first to study diverse representation from the perspective of cohesive groups. We have done so by considering strengthenings of Justified Representation in two distinct directions. First, we considered the known but not well understood notions of Strong and Semi-Strong JR, which consider the same cohesive groups as JR, but require these groups to be represented in a stricter sense. Second, we introduced and studied an axiom called Distinctive Representation, which requires elected candidates to represent groups of voters of equal size. Future research could study, e.g., approximating optimal outcomes for the Cohesive Nash Rule and DR, local variants of the (S)SJR axioms, and more generally, the relation between local optimality and the approximation ratio. 


\bibliography{BIEB}

\begin{thebibliography}{38}
\providecommand{\natexlab}[1]{#1}

\bibitem[{Ai and Tao(2026)}]{ai_strong_2026}
Ai, Y.; and Tao, B. 2026.
\newblock Computational Complexity of Strong and Average Justified Representation.
\newblock arXiv:2606.29643.

\bibitem[{Aziz et~al.(2017)Aziz, Brill, Conitzer, Elkind, Freeman, and Walsh}]{aziz_justified_2017}
Aziz, H.; Brill, M.; Conitzer, V.; Elkind, E.; Freeman, R.; and Walsh, T. 2017.
\newblock Justified Representation in Approval-Based Committee Voting.
\newblock \emph{Social Choice and Welfare}, 48(2): 461--485.

\bibitem[{Aziz et~al.(2018)Aziz, Elkind, Huang, Lackner, {Sanchez-Fernandez}, and Skowron}]{azizComplexityExtendedProportional2018}
Aziz, H.; Elkind, E.; Huang, S.; Lackner, M.; {Sanchez-Fernandez}, L.; and Skowron, P. 2018.
\newblock On the {{Complexity}} of {{Extended}} and {{Proportional Justified Representation}}.
\newblock \emph{Proceedings of the AAAI Conference on Artificial Intelligence}, 32(1).

\bibitem[{Baujard et~al.(2013)Baujard, Gavrel, Igersheim, Laslier, and Lebon}]{baujardGHLL_2013}
Baujard, A.; Gavrel, F.; Igersheim, H.; Laslier, J.-F.; and Lebon, I. 2013.
\newblock Approval voting, evaluation voting.
\newblock Https://doi.org/10.3917/reco.642.0006.

\bibitem[{Baujard et~al.(2025)Baujard, Gavrel, Igersheim, Lebon, and Delemazure}]{baujardGHLD_2025_2012}
Baujard, A.; Gavrel, F.; Igersheim, H.; Lebon, I.; and Delemazure, T. 2025.
\newblock Voter Autrement 2012 - Dataset of the In Situ Experiments.
\newblock Https://zenodo.org/records/15007145.

\bibitem[{Baujard and Igersheim(2010)}]{baujard_2010}
Baujard, A.; and Igersheim, H. 2010.
\newblock Framed field experiments on approval voting: lessons from the 2002 and 2007 french presidential elections.
\newblock In \emph{Handbook of Approval Voting}.

\bibitem[{Baujard, Igersheim, and Delemazure(2025)}]{baujard_2025_2007}
Baujard, A.; Igersheim, H.; and Delemazure, T. 2025.
\newblock Voter Autrement 2007 - Dataset of the In Situ Experiments.
\newblock Https://zenodo.org/records/15007145.

\bibitem[{Bouveret et~al.(2019)Bouveret, Blanch, Baujard, Durand, Igersheim, Lang, Laruelle, Laslier, Lebon, and Merlin}]{bouveret_2019}
Bouveret, S.; Blanch, R.; Baujard, A.; Durand, F.; Igersheim, H.; Lang, J.; Laruelle, A.; Laslier, J.-F.; Lebon, I.; and Merlin, V. 2019.
\newblock Voter Autrement 2017 for the French Presidential Election - The data of the In Situ Experiment.
\newblock Https://zenodo.org/records/15007145.

\bibitem[{Brill et~al.(2025)Brill, Israel, Micha, and Peters}]{brill2025individual}
Brill, M.; Israel, J.; Micha, E.; and Peters, J. 2025.
\newblock Individual representation in approval-based committee voting: M. Brill et al.
\newblock \emph{Social Choice and Welfare}, 64(1): 69--96.

\bibitem[{Caragiannis et~al.(2019)Caragiannis, Kurokawa, Moulin, Procaccia, Shah, and Wang}]{CaragiannisKMPS19}
Caragiannis, I.; Kurokawa, D.; Moulin, H.; Procaccia, A.~D.; Shah, N.; and Wang, J. 2019.
\newblock The Unreasonable Fairness of Maximum Nash Welfare.
\newblock \emph{{ACM} Trans. Economics and Comput.}, 7(3): 12:1--12:32.

\bibitem[{Chamberlin and Courant(1983)}]{chamberlin_representative_1983}
Chamberlin, J.~R.; and Courant, P.~N. 1983.
\newblock Representative {{Deliberations}} and {{Representative Decisions}}: {{Proportional Representation}} and the {{Borda Rule}}.
\newblock \emph{American Political Science Review}, 77(3): 718--733.

\bibitem[{Chlamt{\'{a}}c et~al.(2018)Chlamt{\'{a}}c, Dinitz, Konrad, Kortsarz, and Rabanca}]{ChlamtacDKKR18}
Chlamt{\'{a}}c, E.; Dinitz, M.; Konrad, C.; Kortsarz, G.; and Rabanca, G. 2018.
\newblock The Densest k-Subhypergraph Problem.
\newblock \emph{{SIAM} Journal on Discrete Mathematics}, 32(2): 1458--1477.

\bibitem[{Chlamt{\'{a}}c, Dinitz, and Makarychev(2017)}]{ChlamtacDM17}
Chlamt{\'{a}}c, E.; Dinitz, M.; and Makarychev, Y. 2017.
\newblock Minimizing the Union: Tight Approximations for Small Set Bipartite Vertex Expansion.
\newblock In Klein, P.~N., ed., \emph{Proceedings of the 28th Annual {ACM-SIAM} Symposium on Discrete Algorithms ({SODA})}, 881--899. {SIAM}.

\bibitem[{Chrobak et~al.(2021)Chrobak, Golin, Lam, and Nogneng}]{ChrobakGLN21schedule}
Chrobak, M.; Golin, M.~J.; Lam, T.~W.; and Nogneng, D. 2021.
\newblock Scheduling with gaps: new models and algorithms.
\newblock \emph{Journal of Scheduling}, 24(4): 381--403.

\bibitem[{Damaschke(2017)}]{Damaschke17intervals}
Damaschke, P. 2017.
\newblock Refined algorithms for hitting many intervals.
\newblock \emph{Information Processing Letters}, 118: 117--122.

\bibitem[{Eder, Mochmann, and Quandt(2015)}]{eder_political_2015}
Eder, C.; Mochmann, I.~C.; and Quandt, M. 2015.
\newblock \emph{Political Trust and Disenchantment with Politics: International Perspectives}, volume 125 of \emph{International Studies in Sociology and Social Anthropology}.
\newblock Brill.
\newblock ISBN 978-90-04-26394-9.

\bibitem[{Faliszewski et~al.(2017)Faliszewski, Skowron, Slinko, and Talmon}]{faliszewski2017multiwinner}
Faliszewski, P.; Skowron, P.; Slinko, A.; and Talmon, N. 2017.
\newblock {{Multiwinner Voting: A New Challenge}} for {{Social Choice Theory}}.
\newblock \emph{Trends in computational social choice}, 74(2017): 27--47.

\bibitem[{Feige(1998)}]{feige_threshold_1998}
Feige, U. 1998.
\newblock {{A Threshold}} of ln(n) for {{Approximating Set Cover}}.
\newblock \emph{Journal of the ACM}, 45(4): 634--652.

\bibitem[{Fish et~al.(2024)Fish, G{\"o}lz, Parkes, Procaccia, Rusak, Shapira, and W{\"u}thrich}]{fish2024generative}
Fish, S.; G{\"o}lz, P.; Parkes, D.~C.; Procaccia, A.~D.; Rusak, G.; Shapira, I.; and W{\"u}thrich, M. 2024.
\newblock Generative Social Choice.
\newblock In \emph{Proceedings of the 25th ACM Conference on Economics and Computation}, 985--985.

\bibitem[{Hochbaum and Pathria(1998)}]{hochbaum_analysis_1998}
Hochbaum, D.~S.; and Pathria, A. 1998.
\newblock {{Analysis}} of the {{Greedy Approach}} in {{Problems}} of {{Maximum k-Coverage}}.
\newblock \emph{Naval Research Logistics}, 45(6): 615--627.

\bibitem[{Lackner and Skowron(2023)}]{lackner_multi-winner_2023}
Lackner, M.; and Skowron, P. 2023.
\newblock \emph{Multi-{{Winner Voting}} with {{Approval Preferences}}}.
\newblock {{Springer Briefs}} in {{Intelligent Systems}}. Springer International Publishing.
\newblock ISBN 978-3-031-09015-8 978-3-031-09016-5.

\bibitem[{Landemore(2017)}]{landemore_democratic_2017}
Landemore, H. 2017.
\newblock \emph{{{Democratic Reason: Politics, Collective Intelligence}}, and the {{Rule}} of the {{Many}}}.
\newblock Princeton, NJ Oxford: Princeton University Press, first paperback printing edition.
\newblock ISBN 978-0-691-17639-0 978-0-691-15565-4.

\bibitem[{Lin et~al.(2007)Lin, Yuan, Zhang, and Zhang}]{LinYZZ07skyline}
Lin, X.; Yuan, Y.; Zhang, Q.; and Zhang, Y. 2007.
\newblock Selecting Stars: The k Most Representative Skyline Operator.
\newblock In Chirkova, R.; Dogac, A.; {\"{O}}zsu, M.~T.; and Sellis, T.~K., eds., \emph{Proceedings of the 23rd International Conference on Data Engineering, {ICDE} 2007, The Marmara Hotel, Istanbul, Turkey, April 15-20, 2007}, 86--95. {IEEE} Computer Society.

\bibitem[{Lindeboom et~al.(2026)Lindeboom, Brehm, Grossi, and Murukannaiah}]{lindeboom2025diversecommitteesincompleteinaccurate}
Lindeboom, F.; Brehm, M.; Grossi, D.; and Murukannaiah, P.~K. 2026.
\newblock Diverse Committees with Incomplete or Inaccurate Approval Ballots.
\newblock In \emph{Proceedings of the 25th International Conference on Autonomous Agents and Multiagent Systems}, AAMAS '26, 3704–3712. Richland, SC: International Foundation for Autonomous Agents and Multiagent Systems.
\newblock ISBN 9798400723179.

\bibitem[{Manurangsi(2017)}]{manurangsi2017almost}
Manurangsi, P. 2017.
\newblock Almost-polynomial ratio ETH-hardness of approximating densest k-subgraph.
\newblock In \emph{Proceedings of the 49th Annual ACM SIGACT Symposium on Theory of Computing}, 954--961.

\bibitem[{Manurangsi(2020)}]{manurangsi_tight_2019}
Manurangsi, P. 2020.
\newblock {{Tight Running Time Lower Bounds}} for {{Strong Inapproximability}} of {{Maximum}} \emph{k}-{{Coverage, Unique Set Cover}} and {{Related Problems}} (via \emph{t}-Wise Agreement Testing Theorem).
\newblock In \emph{Proceedings of the 2020 {ACM-SIAM} Symposium on Discrete Algorithms, {SODA} 2020, Salt Lake City, UT, USA, January 5-8, 2020}, 62--81. {SIAM}.

\bibitem[{Mattei and Walsh(2013)}]{MaWa13a}
Mattei, N.; and Walsh, T. 2013.
\newblock PrefLib: A Library of Preference Data \textsc{http://preflib.org}.
\newblock In \emph{Proceedings of the 3rd International Conference on Algorithmic Decision Theory (ADT 2013)}, Lecture Notes in Artificial Intelligence. Springer.

\bibitem[{Mikhaylovskaya(2024)}]{mikhaylovskaya2024enhancing}
Mikhaylovskaya, A. 2024.
\newblock Enhancing {{Deliberation}} with {{Digital Democratic Innovations}}.
\newblock \emph{Philosophy \& Technology}, 37(1): 3.

\bibitem[{Monroe(1995)}]{monroe1995fully}
Monroe, B.~L. 1995.
\newblock Fully proportional representation.
\newblock \emph{American Political Science Review}, 89(4): 925--940.

\bibitem[{Nemhauser, Wolsey, and Fisher(1978)}]{nemhauser_analysis_1978}
Nemhauser, G.~L.; Wolsey, L.~A.; and Fisher, M.~L. 1978.
\newblock {{An Analysis}} of {{Approximations}} for {{Maximizing Submodular Set Functions}}---{{I}}.
\newblock \emph{Mathematical Programming}, 14(1): 265--294.

\bibitem[{Peters, Pierczy{\'n}ski, and Skowron(2021)}]{peters_proportional_2020}
Peters, D.; Pierczy{\'n}ski, G.; and Skowron, P. 2021.
\newblock Proportional participatory budgeting with additive utilities.
\newblock \emph{Advances in Neural Information Processing Systems}, 34: 12726--12737.

\bibitem[{Peters and Skowron(2020)}]{PetersS20}
Peters, D.; and Skowron, P. 2020.
\newblock Proportionality and the Limits of Welfarism.
\newblock In Bir{\'{o}}, P.; Hartline, J.~D.; Ostrovsky, M.; and Procaccia, A.~D., eds., \emph{{EC} '20: The 21st {ACM} Conference on Economics and Computation, Virtual Event, Hungary, July 13-17, 2020}, 793--794. {ACM}.

\bibitem[{Pierczynski and Skowron(2022)}]{PierczynskiS22core}
Pierczynski, G.; and Skowron, P. 2022.
\newblock Core-Stable Committees Under Restricted Domains.
\newblock In \emph{Proceedings of the 18th International Conference on Web and Internet Economics (WINE 2022)}, Lecture Notes in Computer Science, 311--329. Springer.

\bibitem[{Raghavendra and Steurer(2010)}]{raghavendra_graph_2010}
Raghavendra, P.; and Steurer, D. 2010.
\newblock Graph Expansion and the Unique Games Conjecture.
\newblock In \emph{Proceedings of the Forty-Second {{ACM}} Symposium on {{Theory}} of Computing}, 755--764. ACM.
\newblock ISBN 978-1-4503-0050-6.

\bibitem[{{S{\'a}nchez-Fern{\'a}ndez} et~al.(2017){S{\'a}nchez-Fern{\'a}ndez}, Elkind, Lackner, Fern{\'a}ndez, Fisteus, Basanta~Val, and Skowron}]{sanchez-fernandez_proportional_2017}
{S{\'a}nchez-Fern{\'a}ndez}, L.; Elkind, E.; Lackner, M.; Fern{\'a}ndez, N.; Fisteus, J.; Basanta~Val, P.; and Skowron, P. 2017.
\newblock Proportional {{Justified Representation}}.
\newblock \emph{Proceedings of the AAAI Conference on Artificial Intelligence}, 31(1).

\bibitem[{Sornat, {Vassilevska Williams}, and Xu(2022)}]{SornatWX22}
Sornat, K.; {Vassilevska Williams}, V.; and Xu, Y. 2022.
\newblock Near-Tight Algorithms for the Chamberlin-Courant and Thiele Voting Rules.
\newblock In Raedt, L.~D., ed., \emph{Proceedings of the Thirty-First International Joint Conference on Artificial Intelligence, {IJCAI} 2022, Vienna, Austria, 23-29 July 2022}, 482--488. ijcai.org.

\bibitem[{Spaeth et~al.(2023)Spaeth, Epstein, Martin, Segal, Ruger, and Benesh}]{spaeth_2023supreme}
Spaeth, H.~J.; Epstein, L.; Martin, A.~D.; Segal, J.~A.; Ruger, T.~J.; and Benesh, S.~C. 2023.
\newblock 2023 Supreme Court Database (Version 2023 Release 01).
\newblock Http://supremecourtdatabase.org/.

\bibitem[{Szufa et~al.(2022)Szufa, Faliszewski, Janeczko, Lackner, Slinko, Sornat, and Talmon}]{szufa_how_2022}
Szufa, S.; Faliszewski, P.; Janeczko, L.; Lackner, M.; Slinko, A.; Sornat, K.; and Talmon, N. 2022.
\newblock How to {{Sample Approval Elections}}?
\newblock In Raedt, L.~D., ed., \emph{Proceedings of the Thirty-First International Joint Conference on Artificial Intelligence, {IJCAI-22}}, 496--502. International Joint Conferences on Artificial Intelligence Organization.

\end{thebibliography}

\appendix
\onecolumn

\section{Supplementary material to \Cref{sec:ssjr-cc}}
\label{app:relations}

\begin{repproposition}{prop:sjr-score}
We have $\SJR(W)=1 \iff W$ satisfies Strong JR.	
\end{repproposition}
\begin{proof}
For a committee $W$ to satisfy SJR, we require all cohesive groups of voters $N'$ to satisfy the property $W\cap(\cap_{v_i\in N'}A(i)) \neq \emptyset$. In words: at least one candidate that these voters all approve of, needs to be elected. Meanwhile, attaining $\SJR(W)=1$ requires that for all candidates $c$ with $|N_c| \geq \frac nk$, the support groups satisfy this same property: $W\cap(\cap_{v_i\in N_c}A(i)) \neq \emptyset$.

This makes the direction $W$ satisfies SJR $\implies \SJR(W)=1$ obvious: if all cohesive groups of voters $N'$ satisfy $W\cap(\cap_{v_i\in N'}A(i)) \neq \emptyset$, then certainly this is true of the subset of such groups which happen to be all the supporters of some candidate.  

This leaves the direction $\SJR(W)=1\implies W$ satisfies SJR. To see this, suppose instead that $W$ does not satisfy SJR. Then there exists some set of voters $N'\subseteq N$ which is cohesive, but for which $W\cap(\cap_{v_i\in N'}A(i)) = \emptyset$. Writing $c$ for a candidate in $\cap_{v_i\in N'}A(i)$ (which exists because $N'$ is cohesive), consider the support set $N_c$ of $c$. Since all voters in $N'$ support this candidate, we have $N' \subseteq N_c$, so that $W\cap(\cap_{v_i\in N_c}A(i)) = \emptyset$ also. But $|N_c| \geq |N'| \geq \frac nk$ and so $c$ causes $\SJR(W)<1$.


\end{proof}

\begin{repproposition}{prop:ssjr-score}
We have $\SSJR(W)=1 \iff W$ satisfies Semi-Strong JR.	
\end{repproposition}
\begin{proof}
The proof is analogous to the SJR case: simply replace the requirement of $W\cap(\cap_{v_i\in N'}A(i)) \neq \emptyset$ by $N_c \subseteq \cup_{c'\in W} N_{c'}$. 
\end{proof}

\begin{reptheorem}{thm:relations}
\begin{enumerate}
    \item A committee $W$ that satisfies the CC axiom, also satisfies the Semi-Strong JR axiom.
    \item A committee $W$ that satisfies the CC axiom, may not be optimal for Strong JR.
    \item A committee $W$ that satisfies the (Semi-)Strong JR axiom, may not be optimal for CC, and vice versa.
    \item A committee $W$ that is optimal for CC, may not be optimal for Semi-Strong JR.
    \item A committee $W$ that is optimal for  Strong JR, may not be optimal for Semi-Strong JR.
    \item A committee $W$ that is optimal for (Semi-)Strong JR, may not satisfy JR.
    \item A committee $W$ that is optimal for CNASH, satisfies JR
    \item A committee $W$ that is optimal for CNASH, may fail to find a committee that satisfies the SJR axiom, even when it exists.
    \item A committee $W$ that is optimal for CNASH, may not be optimal for CC, and vice versa.
    \item A committee $W$ that satisfies JR, may not be optimal for SSJR max, SJR max, CC max or CNASH.
    \item When satisfying the SSJR axiom is feasible, the set of NASH‑maximizing committees is exactly the set of SSJR satisfying committees.
    \item A committee $W$ that is optimal for CNASH may not be optimal for SSJR max, and vice versa.
    \item A committee $W$ that is optimal for SJR, may not be optimal for CNASH.
\end{enumerate}
\end{reptheorem}
\begin{proof}
\begin{enumerate}
\item If $W$ satisfies the CC axiom, then $|A(i)\cap W|\ge 1$ for all $v_i\in N$, then also $|A(i)\cap W|\ge 1$ for all $i$ in any cohesive group $N'\subseteq N$, which means SSJR is satisfied.
\item Consider the instance in Figure~\ref{fig:prop:cc-sjr}. Here, $n=12$ and suppose $k=4$ so that $n/k=3$. Then all candidates consisting of at least 3 voters are cohesive. By electing committee $W=\{a,b,c,e\}$ we can cover all voters but miss the cohesive group $d$. This committee then satisfies the CC axiom but does not satisfy Strong JR.
\item Consider the instance in \Cref{fig:prop:sjr_ssjr_cc}. Suppose $k=2$, then $n/k=7$. For optimal CC-score we need $W=\{a,b\}$ or $W=\{a,c\}$, but to satisfy (Semi-)Strong JR, we need $W=\{b,c\}$, because $N_b$ and $N_c$ are both cohesive and $N_a$ is not.
\item The instance from the previous item works also for this claim.
\item Consider the instance in Figure~\ref{fig:sjr-ssjr}. Here $n=9$ and suppose $k=3$ so that all four candidates have a cohesive support group. Because all candidates also have a unique support group, all need to be elected to satisfy Strong JR, and any three out of four would attain an optimal solution for Strong JR. However, we can satisfy Semi-Strong JR with $W=\{a,c,d\}$ but not when electing $b$. Hence, $W=\{a,b,c\}$ is optimal for Strong JR but not for Semi-Strong JR.
\item Consider the instance in \Cref{fig:prop:ssjr-jr}. Here, $n=16$ and suppose $k=4$, then $n/k=4$ so that all candidates in this figure have a cohesive support group. Then $W=\{a,b,c,d\}$ is optimal for (S)SJR (all four out of five candidates are) but does not satisfy JR, because then we need candidate $e$.
\item Since $\NASH(W)=0$ only when there exists a cohesive group with zero representatives in $W$, we have $\NASH(W)>0 \iff W$ satisfies JR. A JR committee always exists (e.g. via \textsc{greedy-CC}), so $W$ with $\NASH(W)>0$ will always exist, and CNASH will output $W$ with the largest score, thus satisfying JR.
\item The important difference between CNASH and SJR here is that CNASH only considers the number of covered voters from cohesive groups, and not \emph{how} they are covered. Consider the instance in \Cref{fig:cnash-sjr}. We have $A_1=A_2=\{a,c\},\, A_3=A_4=\{c\},\,A_5=\{a,b\},\,A_6=\{b\}$. When $k=2$, the maximal cohesive groups are: $\{1,2,5\}$ and $\{1,2,3,4\}$. Strong JR thus requires electing $W=\{a,c\}$, then both these cohesive groups are represented by their own candidate. However, also $W=\{b,c\}$ attains optimal NASH score, as it covers all voters, so CNASH will not be able to distinguish between both these committees, hence is not guaranteed to elect an SJR-winner.
\item The important difference between CNASH and CC here is that CNASH is indifferent towards voters of non-cohesive groups. Consider the instance in Figure~\ref{fig:cc-cnash}, and take $k=4$ so that $\lceil n/k\rceil=4$. A CC-winning committee must elect candidate $e$, but this candidate is not cohesive, so CNASH will not elect it.
\item Consider the instance in Figure~\ref{fig:intro}~\subref{fig:intro-a}, where $n=24$ and $k=4$ so that $n/k=6$. Any committee containing $e$ already satisfies JR, but for all other functions, we need to take $W=\{a,b,c,d\}$, and the axioms will even be satisfied.
\item A committee $W$ that maximizes NASH score, satisfies SSJR whenever possible: SSJR is satisfied whenever $\cup_{N'\in\text{Coh}}N'\subseteq \{i:A(i)\cap W\neq\emptyset\}$; all voters of cohesive groups are covered. This means for a committee $W$ that satisfies SSJR, $\NASH(W)=\prod_{N'\in\text{Coh}}|N'|$, the global optimum, which will hence also be outputted by CNASH. Thus, whenever an SSJR satisfying committee $W$ exists, CNASH will output it. A committee $W$ that satisfies the SSJR axiom, is also optimal for CNASH: A committee $W$ that satisfies the SSJR axiom, covers all voters of cohesive groups, so $\NASH(W)=\prod_{N'\in\text{Coh}}|N'|$, which is the maximum the NASH function can attain, therefore also maximizing for this instance.
\item Consider the instance in Figure~\ref{fig:cnash-ssjr} and take $k=2$ so that $n/k=3$. Committee $W=\{b,c\}$ is the NASH winner with score $2^2\cdot4^2\cdot3^7=139968$: we represent voters $1,2,3,5,6$ so that is 2 out of $N_a=\{2,3,4\}$ and 2 out of $N_d=\{1,3,4\}$ but all voters of $N_b$ and $N_c$ and these make many cohesive groups (consisting of only 3 voters that is). This way, only one voter remains unrepresented, but we miss out on covering the two cohesive candidates $a$ and $d$. But this only gives SSJR-score $1/2$, whereas we could attain score $3/4$ with $W=\{a,b\}$ or $W=\{a,c\}$ or $W=\{b,d\}$ or $W=\{c,d\}$. For all these committees, only one cohesive candidate remains uncovered (either $b$ or $c$). However, because $b$ and $c$ are both large candidates, this harms the NASH score, and for all these committees, that is only $3^2\cdot4\cdot3^4\cdot3\cdot2^3=69984$. A last possible committee is $W=\{a,d\}$ but this is for both notions suboptimal. This shows CNASH and SSJR max are incomparable.
\item Consider the instance in \Cref{fig:sjr-to-cnash}, and take $k=2$ so that all candidates form a cohesive voter group. Attaining SJR is not possible here, and any committee attains the same maximal score of $2/3$. However, the only committee that is optimal for CNASH is $W=\{b,c\}$, as it is the only one that  covers all voters.
\end{enumerate}
\end{proof}

\subsection{Figures that accompany the proof of \Cref{thm:relations}}
\begin{figure}[h]\caption{Example instance for claim ``A committee $W$ that satisfies the CC axiom, may not be optimal for Strong JR". There are 12 voters and 5 candidates.}\label{fig:prop:cc-sjr}
\begin{tikzpicture}[scale=0.6, thick]

    \draw (1.5, 3) ellipse (2.2 and 0.7);
    \draw (1.5, 0) ellipse (2.2 and 0.7);

    \draw (0, 1.5) ellipse (0.8 and 1.9);
    \draw (3, 1.5) ellipse (0.8 and 1.9);

    \draw (3, 1.5) ellipse (0.35 and 0.85);

    \foreach \x in {0, 1, 2, 3} {
        \fill (\x, 3) circle (4pt);
        \fill (\x, 0) circle (4pt);
    }

    \foreach \y in {1, 2} {
        \fill (0, \y) circle (4pt);
        \fill (3, \y) circle (4pt);
    }
    
    \node[font=\small] at (-1.2, 2.5) {$a$};
    \node[font=\small] at (1.0, -1.1) {$b$};
    \node[font=\small] at (1.0, 4.1) {$c$};
    \node[font=\small] at (4.0, 2.5) {$d$};
    \node[font=\small] at (3.58, 1.5) {$e$};

\end{tikzpicture}
\end{figure}

\begin{figure}[h]\caption{Example instance for claim ``A committee $W$ that satisfies the (Semi-)Strong JR axiom, may not be optimal for CC, and vice versa". There are 14 voters and 3 candidates.}\label{fig:prop:sjr_ssjr_cc}
\begin{tikzpicture}[thick, scale=0.6]
    \node at (1.6, 2.2) {$b$};

    \draw (0, -0.94) ellipse (1.4cm and 1.52cm);
    \draw (-0.01, 0.47) ellipse (1.4cm and 1.52cm);

    \node at (1.6, -2.2) {$c$};

    \draw (2.85, -0.04) circle (0.8cm);    \node at (4.25, 0.86) {$a$};

    \tikzstyle{element}=[circle, fill=black, inner sep=0pt, minimum size=4pt]

    \node[element] at (-0.35, 1.08) {};
    \node[element] at (-0.05, 1.1) {};
    \node[element] at (0.25, 1.08) {};

    \node[element] at (-0.35, 0.00) {};
    \node[element] at (0.25, 0) {};
    \node[element] at (-0.35, -0.6) {};
    \node[element] at (0.25, -0.6) {};

    \node[element] at (-0.35, -1.45) {};
    \node[element] at (-0.05, -1.45) {};
    \node[element] at (0.25, -1.45) {};

    \node[element] at (2.6, 0.21) {};
    \node[element] at (3.1, 0.21) {};
    \node[element] at (2.6, -0.29) {};
    \node[element] at (3.1, -0.29) {};
\end{tikzpicture}
\end{figure}
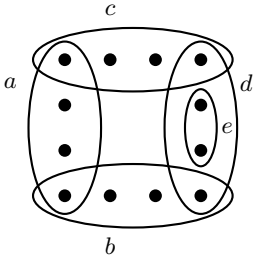

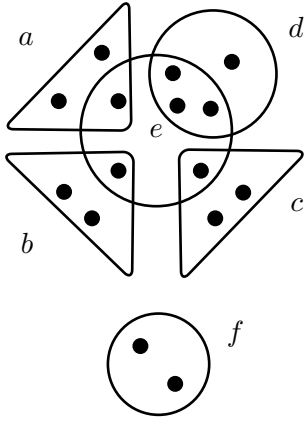
\begin{figure}[h]\caption{Example instance for claim "A committee $W$ that is optimal for Strong JR, may not be optimal for Semi-Strong JR." There are 9 voters and 4 candidates.}\label{fig:sjr-ssjr}
\begin{tikzpicture}[scale=0.8, thick]
\draw (4.24, 0.81) ellipse (0.85 and 0.85);

\draw (3.22, 1.09) ellipse (0.85 and 0.85);
   \draw (2.41, 1.09) ellipse (0.85 and 0.85);
        \draw (5, 1.5) ellipse (0.85 and 0.85);
        \fill (5.3, 1.8) circle (4pt);
        \fill (5.3, 1.2) circle (4pt);
        \fill (4.7, 1.8) circle (4pt);
        \fill (4.7, 1.2) circle (4pt);
        \fill (4.14, 0.51) circle (4pt);
        \fill (2.01, 1.11) circle (4pt);
        \fill (2.95, 0.92) circle (4pt);
        \fill (3.75, 0.97) circle (4pt);
        \fill (2.98, 1.34) circle (4pt);
    
    \node[font=\small] at (1.27, 1.79) {$a$};
    \node[font=\small] at (2.96, 2.24) {$b$};
    \node[font=\small] at (5.01, 0.08) {$c$};
    \node[font=\small] at (6.12, 2.28) {$d$};
\end{tikzpicture}
\end{figure}

\begin{figure}[h]\caption{Example instance for claim ``A committee $W$ that is optimal for (Semi-)Strong JR, may not satisfy JR". There are 16 voters and 5 candidates.}\label{fig:prop:ssjr-jr}
\begin{tikzpicture}[scale=0.6, thick]
    \draw (1.5, 3) ellipse (2.2 and 0.7);
    \draw (1.5, 0) ellipse (2.2 and 0.7);

    \draw (0, 1.5) ellipse (0.8 and 1.9);
    \draw (3, 1.5) ellipse (0.8 and 1.9);

    \draw (5, 1.5) ellipse (0.85 and 0.85);

    \foreach \x in {0, 1, 2, 3} {
        \fill (\x, 3) circle (4pt);
        \fill (\x, 0) circle (4pt);
    }

    \foreach \y in {1, 2} {
        \fill (0, \y) circle (4pt);
        \fill (3, \y) circle (4pt);
    }

        \fill (5.3, 1.8) circle (4pt);
        \fill (5.3, 1.2) circle (4pt);
        \fill (4.7, 1.8) circle (4pt);
        \fill (4.7, 1.2) circle (4pt);
    
    \node[font=\small] at (-1.2, 2.5) {$a$};
    \node[font=\small] at (1.0, -1.1) {$b$};
    \node[font=\small] at (1.0, 4.1) {$c$};
    \node[font=\small] at (4.0, 2.5) {$d$};
    \node[font=\small] at (6.3, 1.5) {$e$};

\end{tikzpicture}
\end{figure}

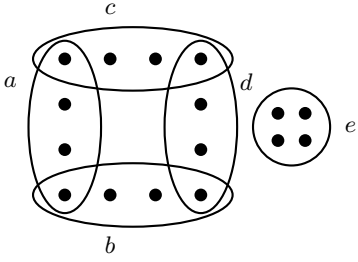
\begin{figure}[h]\caption{Example instance for the claim ``A committee $W$ that is optimal for CNASH, may fail to find a committee that satisfies the SJR axiom, even when it exists." We have $N_a=\{1,2,5\}, N_b=\{5,6\}$ and $N_c=\{1,2,3,4\}$.}\label{fig:cnash-sjr}
\begin{tikzpicture}[scale=0.7, thick]
  \draw (-1.1,0) ellipse (1 and 1.5);
  \draw ( 1.1,0) ellipse (1 and 1.5);
  \draw (0,-1.1) ellipse (2 and 0.9);
  \fill (-1.1, 0.75) circle (4pt);  
  \fill (-1.1,-1) circle (4pt);  
  \fill ( 1.3, 0.75) circle (4pt);  
  \fill ( 0.9, 0.75) circle (4pt);  
  \fill ( 1.3,-1) circle (4pt);  
  \fill ( 0.9,-1) circle (4pt);  

    \node[font=\small] at (-2, 1.5) {$b$};
    \node[font=\small] at (0, -1.6) {$a$};
    \node[font=\small] at (2.0, 1.5) {$c$};

  \node[font=\tiny] at (0.91,-0.69) {1};
  \node[font=\tiny] at (1.28,-0.69) {2};
  \node[font=\tiny] at (0.9,1.05) {3};
  \node[font=\tiny] at (1.28,1.05) {4};
  \node[font=\tiny] at (-1.1,1.05) {6};
  \node[font=\tiny] at (-1.1,-0.69) {5};
\end{tikzpicture}
\end{figure}

\begin{figure}[h]\caption{Example instance for the claim ``A committee $W$ that is optimal for CNASH may not be optimal for SSJR max, and vice versa". There are 6 voters and we have $N_a=\{2,3,4\},\,N_b=\{1,2,3,5\},\,N_c=\{1,2,3,6\},\,N_d=\{1,3,4\}$. }\label{fig:cnash-ssjr}
\begin{tikzpicture}[scale=2]
    \draw[rounded corners=7.1pt] (-0.01,0.01) rectangle (1.85,0.49);
    \fill (0.25, 0.25) circle (0.07cm);
    \fill (0.75, 0.25) circle (0.07cm);
    \fill (1.25, 0.25) circle (0.07cm);
    \draw[rounded corners=7.1pt] (0.58,0.06) rectangle (2.44,0.54);
    \fill (1.75, 0.25) circle (0.07cm);
    \fill (2.25, 0.25) circle (0.07cm);
    \fill (1.25, -0.4) circle (0.07cm);
    \draw [rounded corners] (1.35,-0.67)--(0.35,0.43)--(1.39,0.43)--cycle;
    \draw [rounded corners] (1.11,0.43)--(2.11,0.43)--(1.15,-0.67)--cycle;
    \node[font=\small] at (-0.2, 0.2) {$b$};
    \node[font=\small] at (2.6, 0.3) {$c$};
    \node[font=\small] at (1.0, -0.5) {$a$};
    \node[font=\small] at (1.5, -0.48) {$d$};

  \node[font=\tiny] at (0.75,0.37) {2};
  \node[font=\tiny] at (1.74,0.37) {1};
  \node[font=\tiny] at (1.25,0.37) {3};
  \node[font=\tiny] at (1.25,-0.25) {4};
  \node[font=\tiny] at (2.24,0.37) {6};
  \node[font=\tiny] at (0.25,0.38) {5};
\end{tikzpicture}
\end{figure}

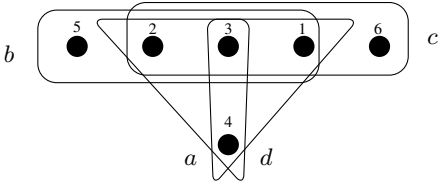
\begin{figure}[h]\caption{Example instance for claim ``A committee $W$ that is optimal for CC, may not be optimal for CNASH, and vice versa". There are 15 voters and 5 candidates.}\label{fig:cc-cnash}
\begin{tikzpicture}[scale=0.6, thick]
    \draw (1.5, 3) ellipse (2.2 and 0.7);
    \draw (1.5, 0) ellipse (2.2 and 0.7);

    \draw (0, 1.5) ellipse (0.8 and 1.9);
    \draw (3, 1.5) ellipse (0.8 and 1.9);

    \draw (5, 1.5) ellipse (0.85 and 0.85);

    \foreach \x in {0, 1, 2, 3} {
        \fill (\x, 3) circle (4pt);
        \fill (\x, 0) circle (4pt);
    }

    \foreach \y in {1, 2} {
        \fill (0, \y) circle (4pt);
        \fill (3, \y) circle (4pt);
    }

        \fill (5.3, 1.8) circle (4pt);
        \fill (5, 1.2) circle (4pt);
        \fill (4.7, 1.8) circle (4pt);
    
    \node[font=\small] at (-1.2, 2.5) {$a$};
    \node[font=\small] at (1.0, -1.1) {$b$};
    \node[font=\small] at (1.0, 4.1) {$c$};
    \node[font=\small] at (4.0, 2.5) {$d$};
    \node[font=\small] at (6.3, 1.5) {$e$};
\end{tikzpicture}
\end{figure}

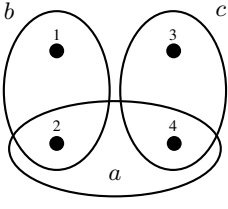
\begin{figure}[h]\caption{Example instance for claim ``A committee $W$ that is optimal for SJR, may not be optimal for CNASH". The profile: $N_a=\{2,4\},\,N_b=\{1,2\},\,N_c=\{3,4\}$.}\label{fig:sjr-to-cnash}
\begin{tikzpicture}[scale=0.7, thick]
  \draw (-1.1,0) ellipse (1 and 1.5);
  \draw ( 1.1,0) ellipse (1 and 1.5);
  \draw (0,-1.1) ellipse (2 and 0.9);
  \fill (-1.1, 0.75) circle (4pt);  
  \fill (-1.1,-1) circle (4pt);  
  \fill ( 1.11, 0.75) circle (4pt);  
  \fill ( 1.12,-1) circle (4pt);  

    \node[font=\small] at (-2, 1.5) {$b$};
    \node[font=\small] at (0, -1.6) {$a$};
    \node[font=\small] at (2.0, 1.5) {$c$};

  \node[font=\tiny] at (-1.1,1.05) {1};
  \node[font=\tiny] at (-1.1,-0.66) {2};
  \node[font=\tiny] at (1.1,1.05) {3};
  \node[font=\tiny] at (1.1,-0.66) {4};
\end{tikzpicture}
\end{figure}

\subsection{Explanation of relations in \Cref{figure:relations}}
Note that there are 12 trivial relations: a committee satisfying any axiom implies that that committee also satisfies the maximization version of that axiom, and the maximization satisfies the axiom whenever possible (6 relations) and the SJR axiom is strictly stronger than the SSJR axiom (2 relations), which is again strictly stronger than the JR axiom (4 relations).

\begin{enumerate}
\item CC axiom $\not\to$ SJR axiom. If otherwise, then CC axiom $\to$ SJR axiom $\to$ SJR max, whereas by \Cref{thm:relations} CC axiom $\not\to$ SJR max.
\item CC axiom $\not\to$ SJR max: \Cref{thm:relations}
\item CC axiom $\to$ SSJR axiom: \Cref{thm:relations}
\item CC axiom $\to$ SSJR max: CC axiom $\to$ SSJR axiom by \Cref{thm:relations} and SSJR axiom $\to$ SSJR max trivially.
\item CC axiom $\to$ JR: known fact, see e.g. \cite{aziz_justified_2017}
\item CC axiom $\to$ CNASH: CC axiom $\to$ SSJR axiom by \Cref{thm:relations} and SSJR axiom $\to$ CNASH also by \Cref{thm:relations}.

\item CC max $\not\to$ SJR axiom: \Cref{thm:relations}
\item CC max $\not\to$ SJR max: If otherwise, then CC axiom $\to$ CC max $\to$ SJR max, whereas by \Cref{thm:relations} CC axiom $\not\to$ SJR max.
\item CC max $\not\to$ SSJR axiom: \Cref{thm:relations}
\item CC max $\not\to$ SSJR max: \Cref{thm:relations}
\item CC max $\to$ JR: known fact, see e.g. \cite{aziz_justified_2017}
\item CC max $\not\to$ CNASH: \Cref{thm:relations}

\item SJR axiom $\not\to$ CC axiom: If otherwise, then SJR axiom $\to$ CC axiom $\to$ CC max, whereas by \Cref{thm:relations} SJR axiom $\not\to$ CC max.
\item SJR axiom $\not\to$ CC max: \Cref{thm:relations}
\item SJR axiom $\to$ SSJR max: SJR axiom $\to$ SSJR axiom $\to$ SSJR max.
\item SJR axiom $\to$ CNASH: SJR axiom $\to$ SSJR axiom by definition and SSJR axiom $\to$ CNASH by \Cref{thm:relations}.

\item SJR max $\not\to$ CC axiom: If otherwise, then SJR max $\to$ CC axiom $\to$ CC max, whereas by \Cref{thm:relations} SJR max $\not\to$ CC max.
\item\label{item:sjr-cc} SJR max $\not\to$ CC max: If otherwise, then SJR axiom $\to$ SJR max $\to$ CC max, whereas by \Cref{thm:relations} SJR axiom $\not\to$ CC max.
\item SJR max $\not\to$ SSJR axiom: If otherwise then SJR max $\to$ SSJR axiom $\to$ SSJR max, whereas by \Cref{thm:relations} SJR max $\not\to$ SSJR max.
\item SJR max $\not\to$ SSJR max: \Cref{thm:relations}
\item SJR max $\not\to$  JR: \Cref{thm:relations}
\item SJR max $\not\to$ CNASH: \Cref{thm:relations}

\item SSJR axiom $\not\to$ CC axiom: If otherwise, then SSJR axiom $\to$ CC axiom $\to$ CC max, whereas by \Cref{thm:relations} SSJR axiom $\not\to$ CC max.
\item SSJR axiom $\not\to$ CC max: \Cref{thm:relations}
\item\label{item:ssjr-sjr} SSJR axiom $\not\to$ SJR max: If otherwise, then CC axiom $\to$ SSJR axiom $\to$ SJR max, whereas by \Cref{thm:relations} CC axiom $\not\to$ SJR max.
\item SSJR axiom $\to$ CNASH: \Cref{thm:relations}

\item SSJR max $\not\to$ CC axiom: If otherwise, then SSJR max $\to$ CC axiom $\to$ CC max, whereas by \Cref{thm:relations} SSJR max $\not\to$ CC max. 
\item SSJR max $\not\to$ CC max: \Cref{thm:relations}
\item SSJR max $\not\to$ SJR axiom: If otherwise, SJR would always be satisfiable.
\item SSJR max $\not\to$ SJR max: If otherwise, then SSJR axiom $\to$ SSJR max $\to$ SJR max, whereas by \Cref{item:ssjr-sjr} SSJR axiom $\not\to$ SJR max. 
\item SSJR max $\not\to$  JR: \Cref{thm:relations}
\item SSJR max $\not\to$ CNASH: \Cref{thm:relations}

\item\label{item:jr-to} JR $\not\to$ SSJR max, SJR max, CC max, CNASH: \Cref{thm:relations}
\item JR $\not\to$ SSJR axiom, SJR axiom, CC axiom: suppose JR $\to$ SSJR axiom, then JR $\to$ SSJR axiom $\to$ SSJR max, whereas that is not true, see \Cref{item:jr-to}. The same holds for the other axioms.

\item CNASH $\not\to$ CC axiom: If otherwise, then CNASH $\to$ CC axiom $\to$ CC max, whereas that is not possible by \Cref{thm:relations}.
\item CNASH $\not\to$ CC max: \Cref{thm:relations}
\item CNASH $\not\to$ SJR axiom: \Cref{thm:relations}
\item CNASH $\not\to$ SJR max: if CNASH $\to$ SJR max, then CNASH $\to$ SJR axiom whenever satisfiable (via SJR max), but that is not the case by \Cref{thm:relations}.
\item CNASH $\to$ SSJR axiom, when satisfiable: \Cref{thm:relations}
\item CNASH $\not\to$ SSJR max: \Cref{thm:relations}
\item CNASH $\to$ JR: \Cref{thm:relations}
\end{enumerate}

\section{Supplementary material to \Cref{sec:ssjr_max}}
\begin{repproposition}{prop:SJR-score} For any approval-based committee election instance and committee $W$, the Semi-Strong JR maximization score and the Strong JR maximization score of $W$ can both be computed in polynomial time.	
\end{repproposition}
\begin{proof}
For SSJR, iterate over all $c\notin W$, check whether $|N_c|\ge \frac{n}{k}$ and if so, check for all $v_i\in N_c$ whether $A(i)\cap W\neq\emptyset$. For SJR we do almost the same, now we keep track of the sets $A(i)\cap W$ for $v_i\in N_c$, in order to compare them.
\end{proof}

\begin{definition}[Chamberlin-Courant decision problem]
In the \textsc{Chamberlin-Courant decision problem (CC)} we are given an instance $I=(N,C,A,k,t)$, where $N$ is a group of $n$ voters, $C$ is a set of $m$ candidates, $A:N\to \mathcal{P}(C)$ is a function that maps every voter to their approval set $A(i)$ of candidates, $k\in\{1,\ldots,m\}$ is the committee size and $t\in \mathbb{N}$. Then $I\in \textsc{CC}\iff \exists W\subseteq C \text{ s.t. } |W|=k \text{ and } |\{i\in N: A(i)\cap W\neq \emptyset \}|\ge t$
\end{definition}

\begin{reptheorem}{thm:SSJR_hard}
The Semi-Strong JR decision problem is NP-complete.
\end{reptheorem}
\begin{proof}
First we show $\text{SSJR}$ is in NP. Consider an instance $I$. A valid witness is a set of candidates $W\subseteq C$, $|W|=k$, such that $|\{c\in C:|N_c|\ge\frac{n}{k},\,\forall v_i\in N_c\,:\,A(i)\cap W\neq\emptyset\}|\ge t$. Let the verifier $\mathbb{M}$ take as input both instance $I$ and witness $W$ and let it determine, for all $c\notin W$, for all $v_i\in N_c$, whether $A(i)\cap W\neq\emptyset$ by going over all $c\in W$, and moreover compute $|N_c|$. This takes time in $O(m^2n)$ and will result in confirmation precisely when $I\in\text{SSJR}$, so SSJR is in NP. It remains to be shown that SSJR is NP-hard. For this, we give a polynomial-time reduction $f$ from instances of CC to instances of SSJR satisfying $I\in \text{CC} \iff f(I)\in\text{SSJR}$.

As is common practice in covering problems, we equate the sets (candidates) with their elements (approving voters). This requires viewing candidates with the exact same set of approving voters as equal, or deleting any one of two such candidates, but electing a second of two candidates with the same set of approving voters has no effect on any of the relevant score functions anyway so in fact we are left with the same instance. To start the reduction, we copy over the set of voters $N$, the set of candidates $C$, the committee size $k$, and the voter profiles $A(i)$ for each $v_i\in N$. We then note that in CC we count the number of voters in the union of all $k$ elected candidates, and ask that of all $n$ voters at least $t$ are in this union. In our new problem, we wish the count the number of large (size at least $\frac nk$) candidates which are contained in the union of all $k$ elected candidates, and ask that of all such large candidates at least $f(t)$ are in this union.
To overcome this difference, we propose to map each voter in the original CC instance to a new candidate. That is, we add $n$ additional candidates $c_{v_i}$; one for each $v_i \in N$. By turning voters into candidates, the idea is that the count of voters in the CC problem translates into counting candidates in the new problem instance. To make this work, these new "voter"-candidates need to be approved by at least $\frac nk$ voters; otherwise they are not cohesive and do not count towards the score. To achieve this, we will introduce $D$ (a number to be specified later) ``dummy" voters, who approve of all candidates: both the $m$ original candidates, as well as the $n$ new candidates. We will choose $D$ so that \textit{all} candidates in the new instance are cohesive in $f(I)$, meaning we can simply count the total number of candidates that are subsets of the $k$ selected candidates, in computing the score in the new instance. We will always represent all the dummy voters, as they approve of every candidate.

This would give us one direction of the proof. Say we set $f(t)=t$. Then for each ``yes" instance - each CC instance where it is possible to select a committee representing at least $t$ voters - we get in the reduced instance a selection of at least $t$ large candidates which are subsets of the union of elected candidates: namely exactly the $t$ candidates $c_{v_i}$ corresponding to the $t$ represented voters $v_i$. 

We would not get the opposite direction however: firstly, because each of the $k$ selected candidates also counts towards fully represented candidates, increasing the score by $k$, and secondly, because there may be overlap within the $k$ originally elected candidates in the CC instance, causing us to represent additional candidates in the new SSJR instance. Thus, in a ``no" instance of CC, where we fail to represent $t$ voters, we also fail to fully represent $t$ of the newly added candidates $c_{v_i}$, but we may represent (many more) other candidates, so that in the end we represent more than $t$ of them in the new SSJR instance. To address the first problem, we can simply set $f(t)=t+k$, which corrects for the $k$ large candidates we elect by default (as we inflate each candidate to be larger than cohesive). To address the second problem, we can add to each of the original candidates $c_j$ a unique dummy voter $v_j$ which approves only of that candidate. This ensures that, even if we elect $k$ of the original candidates, and their union contains potentially more of the original candidates, none of those candidates will be complete subsets of the union as they each have a unique voter.

Now that all newly introduced voters for instance $f(I)$ have been introduced, we can determine $D$. At a first glance perhaps $D=n/k-1$ seems good enough, because then all voter-candidates have $D+1=n/k$ approving voters. This would make them cohesive in instance $I$, but not in instance $f(I)$, because instance $f(I)$ has $D+m$ additional voters, so that also raises the cohesiveness threshold from $n/k$ to $\frac{n+D+m}{k}$. We thus need $D+1\ge \frac{n+D+m}{k}\iff D(k-1)\ge n+m-k\iff D\ge \frac{n+m-k}{k-1}$, so for $k\ge 2$ we can take $D=\lceil\frac{n+m-k}{k-1}\rceil$ to ensure that all voter-candidates in $f(I)$ are cohesive. Assuming $k\ge 2$ is not a problem; for $k=1$, CC is not NP-hard, we just pick the largest candidate. All ``original" candidates are now also cohesive, since they have $|N_{c_j}|+D+1\ge D+1$ approving voters.

Now both directions work. If, in the original CC instance we can represent at least $t$ voters (meaning there exists some selection of $k$ candidates whose union is of size $t$), then under $f$ we obtain an SSJR instance where we can represent $t+k$ cohesive candidates, namely the $t$ candidates we created to correspond to the $t$ represented voters, plus the $k$ ``original" candidates (which are inflated to be cohesive as well). Note that we can assume w.l.o.g. that we indeed would elect $k$ of the original candidates, as electing one of the newly generated ``voter" candidates would never increase the score more (i.e. it adds exactly one more voter to the union of the elected candidates, and there must exist an original candidate that does the same, and may in fact grow the union by more than one). The converse is true as well: if we can represent at least $t+k$ candidates in the SSJR instance obtained under $f$, then by the previous argument we can assume the $k$ elected candidates are all original candidates. Due to each original candidate containing a unique dummy voter, we know that none of the other $m-k$ original candidates are fully represented by the committee, and hence the other $t$ represented candidates must come from ``voter candidates" which implies we can cover $t$ voters in the original CC instance. To summarize, the reduction $f$ guarantees: the original CC instance can be solved with score $t$ if and only if the SSJR instance obtained under the reduction can be solved with score $t+k$. Due to the fact that $f(t)=t+k$, it follows that the CC instance is satisfiable if and only if the SSJR instance is satisfiable, as required. 

All that remains is to argue that $f$ runs in polynomial time. We need to copy $N, C, A$ and $k$, add $D+m$ dummy voters, generate $n$ new candidates of size $D+1$, and finally amend $m$ candidates with one additional voter. All this can clearly be done in polynomial time.
\end{proof}

\begin{lemma}\label{lem:reduction_sjr}
For any CC instance $I$ and any committee $W$ of $f(I)$, we have $\text{SJR}(W)=\text{SSJR}(W)$.
\end{lemma}
\begin{proof}
Since, for any $c$, the statement $\exists d\in W:N_c\subseteq N_d$ implies $N_c\subseteq\cup_{d\in W}N_d$, we have $\{c\in C:|N_c|\ge n/k \text{ and } \exists d\in W:N_c\subseteq N_d\} \subseteq \{c\in C:|N_c|\ge n/k\text{ and } N_c\subseteq\cup_{d\in W}N_d\}$ for all $W$; every candidate $c$ that is SJR-covered by $W$ is also SSJR-covered by $W$, and thus also $\text{SJR}(W)\le\text{SSJR}(W)$. For the other direction we use the properties of $f$ and the instances $f(I)$ creates. Take any candidate $c$ that is SSJR-covered. We will show it is also SJR-covered. If $c$ in $f(I)$ is a voter-candidate $c_{v_i}$ then $N_c$ contains both $v_i$ and the $D$ dummy voters. The fact that $c_{v_i}$ is SSJR-covered means that all its approving voters are covered (we know that all candidates are cohesive), so both $v_i$ and the $D$ dummies are covered. This means there exists some candidate $d\in W$ for which $v_i\in N_d$, but since all $D$ dummy voters approve of all candidates (thus also of $d$), we also automatically have $N_c\subseteq N_d$, so $c$ is SJR-covered. If $c$ is an original candidate $c_j$ that is SSJR-covered, then it must be elected in $W$, because there is a dummy voter that makes it impossible to SSJR-cover otherwise. This automatically means $c$ is SJR-covered. Thus, in any case, $c$ is also SJR-covered and $\text{SJR}(W)\ge\text{SSJR}(W)$, so $\text{SJR}(W)=\text{SSJR}(W)$.
\end{proof}

\begin{reptheorem}{thm:sjr_hardness}
The Strong JR decision problem is NP-complete.
\end{reptheorem}
\begin{proof}
Lemma~\ref{lem:reduction_sjr} shows that the function $f$ from Theorem~\ref{thm:SSJR_hard} is also a reduction from CC to SJR, also with $f(t)=t+k$. SJR is also in NP: Consider an instance $I$. A valid witness is a set of candidates $W\subseteq C$ with $|W|=k$ and such that $|\{c\in C: |N_c|\ge n/k,\,\exists d\in W: N_c\subseteq N_d\}|\ge t$. Let the verifier take as input both instance $I$ and witness $W$ and let it determine, for all $c\notin W$, both $|N_c|$ and (if $|N_c|\ge n/k$) whether there exists $d\in W$ with $N_c\subseteq N_d$. This takes time $O(m^2n)$ and will result in confirmation precisely when $I\in \text{SJR}$, so SJR is in NP.
\end{proof}
\begin{reptheorem}{thm:cnash_hardness}
The CNASH decision problem is NP-hard.
\end{reptheorem}
This result does not follow immediately from Theorem~\ref{thm:SSJR_hard} because the SSJR-decision problem is formulated with a threshold $t$ (not existence of a committee that satisfies the SSJR axiom). We therefore need a separate reduction. The proof uses a reduction from Exact 3-Cover, which we define below as a decision problem.
\begin{definition}[Exact 3-Cover Decision Problem]
In the \textsc{Exact 3-Cover} (X3C) decision problem we are given an instance $I=(U,\mathcal{S})$, where $U$ is a set of $3q$ elements, for some $q\in\mathbb{N}$, and $\mathcal{S}$ is a set of 3-element subsets of $U$. Then $I\in\textsc{X3C}\iff \exists S\subseteq \mathcal{S}$ of $q$ sets $s_i$, for which $\cup\,S = U\,\land s_i\cap s_j=\emptyset\, \forall\, s_i,s_j\in S$.
\end{definition}

\begin{proof}
Consider any instance $I=(U,\mathcal{S})$ of X3C with $|U|=3q$ for some $q\in\mathbb{N}$ and $\mathcal{S}$ a family of 3-element subsets. We assume every $u\in U$ lies in some set $S\subset \mathcal{S}$ (so that all elements can be covered) and $|\mathcal{S}|\ge q$ (so that $k\le m$). We map this to an instance of CNASH using function $f$: $f$ maps elements $u\in U$ to voters $v_i$, sets $S\in\mathcal{S}$ to candidates $c_i$, keeping the relations $u\in S \iff v_i\in c_i$, and maps $q$ to $k$. Then $n/k=3q/q=3$. This can all clearly be done in polynomial time. Lastly, we take $t=3^{|\text{Coh}|}$. Since every candidate has exactly 3 approving voters, this means all candidates are cohesive and Coh is the set of all (distinct) sets $N_c$, thus  $|\text{Coh}|\le|\mathcal{S}|$, so $f$ runs in polynomial time. The result is, that $\NASH(W)\ge t$ precisely when all voters are covered, because each cohesive group is size precisely 3. Since $W$ consists of $q$ sets of
size three, it covers all $3q$ voters iff it is an exact cover. This immediately shows both directions, and we conclude $I\in\textsc{X3C}$ iff $f(I)\in\textsc{CNASH}$.
\end{proof}

\begin{reptheorem}{thm:SJR_opt_approx}
Electing $k$ candidates, no polynomial time algorithm can approximate the optimal solution for (S)SJR with ratio better than $1-1/\e$, unless P = NP.
\end{reptheorem}
\begin{proof}
The proof follows from the reduction $f$ used to prove NP-completeness. Suppose that we have a polynomial-time algorithm $\mathcal{A}$ for SSJR that guarantees, in terms of SSJR-score, an approximation ratio of $1-1/e+\epsilon$ for some $\epsilon>0$. We will derive that P = NP. Let $I$ be some instance of CC with $n$ voters. Let us write $\opt_{\text{CC}}(I)$ to denote the optimal absolute CC score achievable on this instance and $\text{CC}_{I}(W)$ to denote the absolute CC score achieved by some solution $W$ on $I$.

Recall that the reduction $f$ has the property that if $W$ is some solution for $I$, then $\text{CC}_{I}(W)+k=\SSJR_{f(I)}(W)$, where the latter denotes the absolute SSJR-score achieved by the same solution $W$ on the SSJR instance $f(I)$. Note that this also implies that $\opt_{\text{CC}}(I)+k=\opt_{\text{SSJR}}(f(I))$, where $\opt_{\text{SSJR}}(f(I))$ is the optimal absolute SSJR-score achievable on instance $f(I)$.

Now, use the reduction $f$ to obtain an instance $f(I)$ of SSJR, and use $\mathcal{A}$ to obtain a solution $W$ such that $\SSJR_{f(I)}(W) \ge \alpha\cdot\opt_{\text{SSJR}}(f(I))$, where $\alpha>1-1/\e$. We can now use the above two equalities to translate this into statements about $\text{CC}(W)$ and $\opt_{\text{CC}}(I)$. To do so, we write $\SSJR_{f(I)}(W) = \text{CC}_{I}(W)+k$, and likewise we write $\alpha \opt_{\text{SSJR}}(f(I)) = \alpha (\opt_{\text{CC}}(I)+k)$. Combining these two tells us that $\text{CC}_I(W)+k \ge \alpha(\opt_{\text{CC}}(I)+k)$, i.e. $\text{CC}_I(W) \ge \alpha(\opt_{\text{CC}}(I)+k) - k$.


Now write $\alpha=1-1/\e+\epsilon$ for a constant $\epsilon>0$. \citet{feige_threshold_1998} proves that, for any $\epsilon'>0$, it is NP-hard to distinguish CC instances in which $k$ candidates can cover all $n$ voters from instances in which no $k$ candidates cover more than $(1-1/\e+\epsilon')\,n$ voters. We take $\epsilon'=\epsilon/4$. 
Given an instance $I$, transform it into instance $I_T$ that is derived from $I$ by making $T=\lceil 2k/\epsilon\rceil$ copies of each voter (keeping approval information the same). This takes polynomial time, and by increasing the number of voters from $n$ to $nT$, additive values of $k$ will cease to affect the approximation ratio.
Run $\mathcal{A}$ on $f(I_T)$ to obtain $W$. On instances of
the first kind, $\opt_{\text{CC}}(I)=n$, so $\opt_{\text{CC}}(I_T)=nT$, and
$$\text{CC}_{I_T}(W)\ge\alpha nT-(1-\alpha)k\ge(1-1/\e+\epsilon)\,nT-(\epsilon/2)\,nT
=(1-1/\e+\epsilon/2)\,nT,$$
using $(1-\alpha)k\le k\le(\epsilon/2)\,nT$ for the second inequality.
On instances of the second kind, $\text{CC}_{I_T}(W)\le(1-1/\e+\epsilon/4)\,nT$. Hence the
algorithm distinguishes the two cases in polynomial time, which implies P = NP.

The result for SJR follows via Lemma~\ref{lem:reduction_sjr}.
\end{proof}

\begin{reptheorem}{thm:ratio}
Algorithm~\ref{alg:greedy2-ssjr} is $(1-1/\e)$-approximate for SJR and runs in polynomial time. 
\end{reptheorem}

\begin{proof}
The proof is similar to the well-known proof of the same $1-1/\e$ approximation ratio for the greedy algorithm for CC \cite{nemhauser_analysis_1978}. 
In each of the $k$ rounds, the greedy algorithm picks the candidate which increases the CC score most. The key idea in that proof is the following. At the start of the $\ell$'th iteration (out of $k$ in total), the difference between the score of the currently generated committee (of size $\ell-1$) and the optimal committee (of size $k$) is $\opt - \gr(\ell-1)$. This means that there are at least that many voters that are covered in the (fixed) optimal solution, but not yet covered by the current committee of size $\ell-1$. 

By the pigeonhole principle, there exists a candidate in the optimal solution, not yet selected in the current committee, that will cover at least $(\opt - \gr(\ell-1))/k$ of these elements. This means that in iteration $\ell$, the greedy algorithm increases the CC score by at least this amount. That is, $\gr(\ell) \geq \gr(\ell-1) + \frac{\opt - \gr(\ell-1)}{k}$. Writing out this sequence starting from $k$, ultimately yields the bound $\gr(k) \geq \opt (1-(1-\frac 1k)^k) \geq \opt(1-\frac{1}{\e})$. See e.g. \citet{nemhauser_analysis_1978} for the original proof, or \citet[lemma 10]{lindeboom2025diversecommitteesincompleteinaccurate} for an explicit derivation.

We can copy this reasoning for our greedy algorithm for SJR, because the proof only relies on the fact that in round $\ell$ the CC score is increased by at least $(\opt - \gr(\ell-1))/k$. It is not hard to see that this also holds for SJR, i.e. in round $\ell$ the SJR-score is increased by at least this amount, by the same pigeonhole principle argument as for CC. We can thus re-use the derivation for the CC case to obtain the same $1-1/\e$ approximation ratio. Note that this reasoning does not hold for SSJR, because the SSJR objective is not submodular (i.e. there exist $W_1,\,W_2\subseteq C$ such that $\text{SSJR}(W_1)+\text{SSJR}(W_2)<\text{SSJR}(W_1\cup W_2)+\text{SSJR}(W_1\cap W_2)$).

As for the running time, we do $k$ rounds. In each round, we iterate over all candidates and compute the increase in score they would generate. We noted already in \Cref{sec:ssjr_max} that the SJR-score can be computed in polynomial time. This means we can also compute the difference in score between the two committees in polynomial time. The resulting complexity is therefore also polynomial time.
\end{proof}

\begin{definition}
In the \textsc{Densest-$k$ Subgraph problem} (DkS), we are given an instance $I=(G,k,t)$, where $G=(U,E)$ is a graph with node set $U$ and edge set $E$, and $k\le|U|$ is a number, $t\in\mathbb{N}$. Then $I\in \text{DkS}\iff \exists S\subseteq U$ s.t. $|S|=k$ and $|E(S)|\ge t$.
\end{definition}

\begin{theorem}\label{thm:ssjr-dks-hardness}
If, for some $\alpha\in(0,1]$, there exists a polynomial-time $\alpha$-approximation algorithm for Semi-Strong JR maximization, then there exists a polynomial-time $\frac{\alpha}{2}$-approximation algorithm for Densest-k Subgraph.
\end{theorem}
\begin{proof}
We will assume that we have an algorithm $\mathcal{A}$ that is $\alpha$-approximate for SSJR and we will show that this implies an $\alpha\cdot\frac{|U|-1}{|U|}$-approximate algorithm for DkS must exist, then conclude by contraposition, using $\frac{|U|-1}{|U|}=1-\frac{1}{|U|}\ge\frac{1}{2}$ because $|U|\ge 2$. For this, we will construct a function $f$ that maps any DkS-instance $(G,k)$ into an SSJR-instance $f((G,k))$ (on which we can then run $\mathcal{A}$ to find an $\alpha$-approximate solution $W$) and a function $g$ that maps any such solution $W$ for $f((G,k))$ back into a solution $g(W)$ for DkS on $(G,k)$. For this reasoning to work, we show the functions satisfy (for $M$ defined later so that the ratio holds true)
\begin{itemize}
    \item[(i)] $\opt_{\text{SSJR}}(f((G,k)))= k +M\cdot\opt_{\text{DkS}}((G,k))$, and
    \item[(ii)] $\#\text{SSJR}(W)\le k+M\cdot\text{DkS}(g(W))$ for any solution $W$ of $f(G,k)$.
\end{itemize}
Let $G=(U,E)$ be a graph with nodes $U$ and edges $E$. Together with parameter $k$ this forms the DkS instance $(G,k)$. Without loss of generality, we may assume $k\ge 2$ and $E\neq\emptyset$ so that the optimal value for DkS on the instance is greater than 0. Moreover, $G$ has no isolated vertices and $|U|\ge k$.

The function $f$ is defined as follows. For each $u\in U$ we define one 'vertex candidate' $c_u$ and for each $e\in E$, we define $M=k\cdot |U|$ `parallel candidates' $c_e^1,\ldots, c_e^M$. For each edge $e=\{u,v\}$ and each $j\in\{1,\ldots,M\}$ we have one voter $x^j_{e,u}$ that approves of both $c_u$ and $c_e^j$ and one voter $x^j_{e,v}$ that approves of both $c_v$ and $c_e^j$. Lastly, we introduce $L=2M|E|$ voters that approve of all candidates. We map the number of vertices $k$ to a committee size of also $k$. Now note that the total number of voters is $L+2M|E|=4M|E|$ and because $k\ge2$ this means $n/k\le 2M|E|=L$, which means that all candidates are cohesive. 

We define the function $g$ as follows. Given a committee $W$ (i.e. an output of SSJR), take $S=\{u:c_u\in W\}$ and output whichever of the following two sets has more edges: (i) take $S$ and add arbitrary vertices to create set $S'$ of size $k$ or (ii) create an arbitrary set of $k$ vertices with at least one edge between them and call it $S'$, and so we write $g(W)=S'$.

We will now prove (i), starting  with ``$\le$". Fix any committee $W$ and denote $S=\{u:c_u\in W\}$ and write $F$ for all the parallel candidates that are in $W$, so that $k=|S|+|F|$. A voter $x_{e,u}^j$ approves of only $c_u$ and $c_e^j$ and $x_{e,v}^j$ approves of only $c_v$ and $c_e^j$, so in order to cover candidate $c_e^j$ we need to elect this candidate $c_e^j$ ($c_e^j\in F$) or elect both $u$ and $v$ ($u,v\in S$). But in order to cover $c_u$, we need to elect $c_u$, we cannot cover it otherwise. Indeed, if $c_u\notin W$, each voter $x_{e,u}^j\in c_u$ can only be covered by candidate $c_e^j$, but then we need to elect all $M\cdot \text{deg}(u)\ge M>k$ of these parallel candidates, which is not possible. With this knowledge we can rewrite the number of covered candidates for SSJR, that we will denote by $\#\text{SSJR}(W)$ (mind the difference with $\text{SSJR}(W)$), as
\[
\#\text{SSJR}(W) = |S| + M\cdot|E(S)|+|\{c_e^j\in F:e\notin E(S)\}|.
\]
Because $|S|+|F|=k$ and $|\{c_e^j\in F:e\notin E(S)\}|\le |F|$, this reduces to 
\begin{align}\label{eq:ssjr-ratio-1}
\#\text{SSJR}(W)\le k + M\cdot |E(S)|,
\end{align} where equality holds when $|S|=k$ and thus $F=\emptyset$. Now take any committee $W$ for SSJR and use function $g$ to turn it into a solution $S'$ for DkS. Then $|E(S)|\le |E(S')|\le \opt_{\text{DkS}}((G,k))$. Then by Equation~\ref{eq:ssjr-ratio-1}, we have $\#\text{SSJR}(W)\le k + M\cdot \opt_{\text{DkS}}((G,k))$. Since $W$ was an arbitrary committee, this inequality holds also for the optimal committee, hence $\opt_{\text{SSJR}}(f(G,k))\le k + M\cdot \opt_{\text{DkS}}((G,k))$. This proves ``$\le$". For $``\ge"$, let $S^*$ be an optimal solution in instance $(G,k)$, for DkS, and take the committee $W^*=\{c_u:u\in S^*\}$. Because $W^*$ contains only vertex candidates, Equation~\ref{eq:ssjr-ratio-1} holds with equality, so $\#\text{SSJR}(W^*)=k+M\cdot\opt_{\text{DkS}}((G,k))$ and thus $\opt_{\text{SSJR}}(f(G,k))\ge k+M\cdot\opt_{\text{DkS}}((G,k))$. This proves (i).

For (ii), first note that $|E(g(W))|\ge |E(S')|$ because we attained $S'$ simply by adding vertices, and $|E(S')|\ge\frac{\#\text{SSJR}(W)-k}{M}$ by Equation~\ref{eq:ssjr-ratio-1}, so this shows $|E(g(W))|\ge \frac{\#\text{SSJR}(W)-k}{M}$, which is equivalent to $\#\text{SSJR}(W)\le \text{DkS}(g(W))\cdot M + k$.

Now that we have both maps $f$ and $g$ that satisfy both requirements (i) and (ii), we can put everything together to attain the claimed result. From $(G,k)$, employ $f$ to get $f((G,k))$ and run $\mathcal{A}$ to obtain $W$ that is $\alpha$-approximate for SSJR. Then run $g$ on $W$ to obtain $g(W)$. Then
\[
|E(g(W))|\ge \frac{\#\text{SSJR}(W)-k}{M}\ge \frac{\alpha\opt_{\text{SSJR}}-k}{M}=\frac{\alpha(k+M\cdot\opt_{\text{DkS}})-k}{M}=\frac{(\alpha-1)k}{M}+\alpha\opt_\text{DkS}\ge\alpha\opt_\text{DkS}-\frac{1}{|U|}
.\]
We distinguish two cases. If $\alpha\cdot\opt_{\text{DkS}}\ge 1$, then $\frac{1}{|U|}\le\frac{\alpha\cdot\opt_{\text{DkS}}}{|U|}$, so $|E(g(W))|\ge\alpha(1-\frac{1}{|U|})\opt_{\text{DkS}}$. If $\alpha\cdot\opt_{\text{DkS}}<1$, then, since $g(W)$ contains at least one edge by construction, $|E(g(W))|\ge 1>\alpha\cdot\opt_{\text{DkS}}\ge\alpha(1-\frac{1}{|U|})\opt_{\text{DkS}}$. In both cases $|E(g(W))|\ge\alpha(1-\frac{1}{|U|})\opt_{\text{DkS}}\ge\frac{\alpha}{2}\opt_{\text{DkS}}$, because $|U|\ge 2$.
This completes the proof.
\end{proof}

\begin{reptheorem}{thm:ssjr-approxratios}
(i) Under the Small Set Expansion Hypothesis (SSEH), no polynomial-time constant-factor approximation-algorithm exists for SSJR. (ii) Under ETH, there exists a constant $c$ such that there exists no $m^{-1/(\log\log m)^c}$-approximation algorithm for SSJR and there exists a constant $d$ such that there exists no $n^{-1/(\log\log n)^d}$-approximation algorithm for SSJR.
\end{reptheorem}
\begin{proof}
(i) This is a direct consequence of \citet{raghavendra_graph_2010} in combination with Theorem~\ref{thm:ssjr-dks-hardness}.
(ii) By \citet{manurangsi2017almost}, under ETH, there exists a constant $c_0>0$ such that no polynomial-time algorithm is $N^{-1/(\log\log N)^{c_0}}$-approximate for DkS on graphs with $N$ vertices. Let $c=c_0+1$ and suppose some polynomial-time algorithm is $m^{-1/(\log\log m)^{c}}$-approximate for SSJR-maximization. In $f((G,k))$, the number of candidates $m$ is polynomial in $|U|$: $m=|U|+k|U|\cdot|E|\le|U|^4$. Thus, $m^{-1/(\log\log m)^{c}}\ge(|U|^4)^{-1/(\log\log |U|^4)^{c}}$, because $x\mapsto x^{-1/(\log\log x)^{c}}$ is decreasing for sufficiently large $x$. Theorem~\ref{thm:ssjr-dks-hardness} yields a polynomial-time algorithm for DkS with approximation ratio at least
\[
\tfrac12\, m^{-1/(\log\log m)^{c}}\ \ge\ \tfrac12\,|U|^{-4/(\log\log|U|)^{c}}\ >\ |U|^{-1/(\log\log|U|)^{c_0}}
\]
for all sufficiently large $|U|$ (depending on $c_0$). This contradicts \citet{manurangsi2017almost}. The argument for $n$ is identical with $d=c_0+1$, using $n=2\cdot2M\cdot |E|=4k|U|\cdot |E|\le 2|U|^4\le|U|^5$.
\end{proof}

\begin{remark}
Note that, adding $f(t)=k+Mt$, the function $f$ also provides a polynomial-time reduction from DkS to SSJR. If $I=(G,k,t)\in \text{DkS}$ that means there exists $S\subseteq U$ with $|S|=k$ such that $|E(S)|\ge t$, then the committee $W=\{c_u:u\in S\}$ consists of only vertex candidates, thus Equation~\ref{eq:ssjr-ratio-1} holds with equality and $\#\text{SSJR}(W)= k+M|E(S)|\ge f(t)$. Conversely, if $f((G,k))\in\text{SSJR}$, i.e. there exists a committee $W$ with $\#\text{SSJR}(W)\ge f(t)$, then $|E(g(W))|\ge (f(t)-k)/M=t. $
\end{remark}

\section{Supplementary material to \Cref{sec:domains}}\label{app:domains}
\begin{definition}[Candidate Interval (CI)]
    Given an election $E = (N,C,\boldsymbol{A},k)$, we say that $E$ has \emph{candidate interval (CI)} preferences if there exists a linear order $\sqsupset_{C}$ over candidates $C$ such that for every voter $i$, the approval ballot $A(i)$ forms an interval in the ordering $\sqsupset_{C}$.
\end{definition}

\begin{definition}[Candidate Extremal Interval (CEI)]
    Given an election $E = (N,C,\boldsymbol{A},k)$, we say that $E$ has \emph{candidate extremal interval (CEI)} preferences if there exists a linear order $\sqsupset_{C}$ over candidates $C$ such that for every voter $i$, both $A(i)$ and $C\setminus A(i)$ form intervals in the ordering $\sqsupset_{C}$.
\end{definition}

\begin{definition}[Voter Interval (VI)]
    Given an election $E = (N,C,\boldsymbol{A},k)$, we say that $E$ has \emph{voter interval (VI)} preferences if there exists a linear order $\sqsupset_{N}$ over voters $N$ such that for each candidate $c\in C$, the voters who approve of $c$ form an interval in the ordering $\sqsupset_{N}$
\end{definition}

\begin{definition}[Voter Extremal Interval (VEI)]
    Given an election $E = (N,C,\boldsymbol{A},k)$, we say that $E$ has \emph{voter extremal interval (VEI)} preferences if there exists a linear order $\sqsupset_{N}$ over voters $N$ such that for each candidate $c\in C$, both the voters who approve of $c$, and those that disapprove of $c$, form intervals in the ordering $\sqsupset_{N}$.
\end{definition}

\citet{brill2025individual} showed that SSJR is always satisfiable on CEI and VEI domains but the same does not hold for CI and VI domains. The latter fact also means that SJR cannot always be satisfied on CI and VI domains.


\begin{reptheorem}{thm:sjr-domains}
SJR is always satisfiable on the VEI domain, but not on the CEI domain.
\end{reptheorem}

\begin{proof}
\begin{itemize}
    \item[\Checkmark] \label{(s)sjr-vei}
Take a VEI order. For $k=1$, any cohesive candidate is unanimously approved so simply elect such a candidate. For $k\geq 2$, electing at least the following two candidates satisfies SJR: the two most approved candidates that approved by the leftmost and rightmost voter in the VEI order, respectively. This way, if cohesive candidates even exists, you will pick both two that together contain all other (cohesive) candidates. So all voters that are in a cohesive candidate, are represented (by the same candidate).
    \item[\XSolidBrush] \label{sjr-ci-vi-cei}
The following instance has domain restrictions that make it CI, VI and CEI simultaneously: $A_{1} = \{a\},A_{2} = \{a,b\},A_{3} = \{a,b,c\},A_{4} = \{b,c,d\},A_{5} = \{c,d\},A_{6} = \{d\}$. With $k=2$, each candidate needs to be in the committee.
\end{itemize}
\end{proof}

\begin{reptheorem}{thm:(S)SJR-max-CI-VI-poly}
Maximizing the (S)SJR-score can be efficiently done on the CI and VI domains.
\end{reptheorem}
\begin{proof}
We begin with SJR maximisation. First, we set $C' = C\setminus\{c: |N_{c}|<\nicefrac{n}{k}\}$ to restrict ourselves to cohesive candidates. 
We say a cohesive candidate $c$ \emph{SJR-covers} a cohesive candidate $d$ if $ c\in \bigcap_{v_i\in N_{d}}A(i)$.

\paragraph{SJR-max on CI domain:} For each cohesive candidate $c$, we consider all the voter ballots that contain $c$. With all such ballots being intervals in the CI ordering $\sqsupset_{C}$, we consider their left and right endpoints. We then consider the \emph{candidate interval} $I^{CI}(c) = [l^c,r^c]$ for candidate $c$ where:
\begin{itemize}
        \item $l^c$ is the rightmost of all the left endpoints of ballots containing $c$.
        \item $r^c$ is the leftmost of all the right endpoints of ballots containing $c$.
    \end{itemize}
    Observe that a cohesive candidate $c$ SJR-covers a cohesive candidate $d$ if and only if $c\in I^{\text{CI}}(d)$. The SJR maximization problem is now to select $k$ committee members so as to maximize the number of candidate intervals that are intersected with. Note that this is exactly the \textsc{Partial Interval Hitting Set} problem \cite{Damaschke17intervals}, which can be solved in $O(km^2)$ time \cite{ChrobakGLN21schedule}.

\paragraph{SJR-max on VI domain:} For each cohesive candidate $c$, we write $p^{VI}(c) = (l^{c},-r^{c})$ where $l^{VI}(c)$ and $r^{VI}(c)$ are the indices (in the VI ordering) of the leftmost and rightmost voter that approves of $c$, respectively. We say that a cohesive candidate $c$ \emph{weakly dominates} a cohesive candidate $d$ if $p^{VI}(c) \leq p^{VI}(d)$ which holds when both $l^{VI}(c) \leq l^{VI}(d)$ and $-r^{VI}(c) \leq -r^{VI}(d)$ hold. Note that a candidate $c$  SJR-covers a candidate $d$ if and only if $c$ weakly dominates $d$. Thus, maximizing the SJR score is equivalent to selecting $k$ committee members such that the number of candidates that are weakly dominated by some committee member is maximized. Thus, the problem is equivalent to the $2$-dimensional case of selecting $k$-representative \emph{skylines} \cite{LinYZZ07skyline} where our weak dominance property is used instead of their strict version (where at least one of $l^{VI}(c) < l^{VI}(d)$ and $-r^{VI}(c) < -r^{VI}(d)$ must hold). 
A $O(km^2)$ time dynamic programming algorithm can be used for it \cite{LinYZZ07skyline}.

Now, we move onto the SSJR maximisation problem. For the following reformulations of SSJR maximisation, note that there are related problems that have been studied, namely \emph{Densest $k$-Subhypergraph} \cite{ChlamtacDKKR18} and \emph{Minimal $k$-Union} \cite{ChlamtacDM17}

\paragraph{SSJR-max on CI domain:} On the CI domain
, the approval ballots $A(i)$ are intervals over the CI ordering $\sqsupset_{C}$, so we refer to them as approvals intervals.

Let $W$ be a selection of (at most) $k$ candidates. A voter is covered if and only if their approval interval intersects with $W$.

Note that a candidate $c$ is not SSJR-covered by $W$ if and only if there exists some supporter $i$ in $N_c$ such that $A(i)$ does not intersect with $W$. So the set of cohesive candidates that are not SSJR-covered by $W$ are: \[
\mathrm{U}(W) = \{ c\in C \mid |N_c|\geq \frac{n}{k}, A(i)\cap W = \emptyset \text{ for some } i\in N_c\}.
\]

Thus, SSJR-maximization is equivalent to minimising $|\mathrm{U}(W)|$.

Now, take a committee of candidates (ordered as in their CI order) $W=\{c_1 \sqsupset_{C} \dots \sqsupset_{C} c_k\}$. Then add dummy candidates $c_0$ and $c_{k+1}$ so that $W = \{c_0 \sqsupset_{C} c_1 \sqsupset_{C} \dots \sqsupset_{C} c_k \sqsupset_{C} c_{k+1}\}$. Observe that a voter's approval interval does not intersect with $W$ if and only if $A(i)$ is contained in a gap interval $(c_{t},c_{t+1})$, i.e, after the right endpoint of $c_{t}$ and left endpoint of $c_{t+1}$. So, we get that:
\[
|\mathrm{U}(W)| = \sum_{t=0}^{k} w(c_t, c_{t+1})
\]
where $w(c_x,c_y)$ is the number of cohesive candidates that have at least one supporter whose approval interval lies in the gap $[c_{x+1}, c_{y-1}]$. These $w(c_x,c_y)$ can be computed in $O(m^2 n)$ where we do the following: we fix a left endpoint $l =c_x$; iterate through right endpoints $r\in\{c_x,\ldots, c_{k+1}\}$; identify the approval intervals $A(i) = [a,b]$ such that $c_x \sqsupset_{C} a$ and $b \sqsupset_{C} c_y -1$ given the CI order ; and for each $[c_x,r]$, we maintain a size of the union of approval intervals that it contains. Note that there is no double counting as candidates can only be uncovered by the gap that contains it.

Here is the dynamic program. For $j\in [m]$ and $d\in[k]$, $M(j, d)$ be the minimum number of candidates that are not SSJR-covered by the candidates to the left of candidate $c_j$ in the CI order, given that $c_j$ is the $d$-the selected candidate with $M(0,0) = 0$:

\[
M(j,d) = \min_{0 \leq g < j} \; M(g, d-1) + w(g, j) 
\]

The answer of the DP is $|\{c\in C\mid |N_c|\geq \nicefrac{n}{k}\}| - \min_{g\in[m]} \big(M(g,k) + w(g, m+1)\big)$. Note that if the DP selects fewer than $k$ candidates, we can always add candidates to the committee without increasing the number of candidates that are not SSJR-covered.

The DP has $O(mk)$ entries with each taking $O(m)$ time to compute, thus, this DP has a running time of $O(m^2 k)$.

\paragraph{SSJR-max on VI domain:} On the VI domain, the supporter sets $N_c$ of candidates are intervals over the VI ordering of $N$ so we refer to them as supporter intervals. Take the voters in the VI order $1, \dots, n$, and also take the collection of supporter intervals $S=\{N_c\mid c\in C\}$. 

Now, for a set of voters $U \subseteq [n]$, we write $c(U)$ to denote the number of supporters intervals of $S$ contained in $U$, i.e., the number of cohesive candidates that are SSJR-covered by $U$.

For a set of voters $U \subseteq [n]$, we can decompose it into \emph{blocks}, i.e., $U = [a_1, b_1] \cup \ldots \cup [a_t, b_t]$ where a block $[a, b]$ of $U$ is maximal interval of voters in $U$ with the blocks being pairwise disjoint and separated by at least one voter not in $U$, i.e., they are non-adjacent. 

For a block of voters $[a,b]$, we also compute the following:
\begin{itemize}
    \item \emph{Block weight} $w(a,b)$ is the number of supporter intervals $N_c$ of $S$ contained in $[a,b]$ with $|N_c|\geq \nicefrac{n}{k}$ (i.e., the number of cohesive candidates that are SSJR-covered by $[a,b]$). This can be computed using DP similar to the $2D$ prefix sum.
    \item \emph{Block cost} $\mathrm{cost}(a,b)$ is the fewest supporter intervals of $S$ contained in $[a,b]$ whose union is exactly $[a,b]$, or $\mathrm{cost}(a,b)=\infty$ if there is no such selection of supporter intervals in $[a,b]$. (i.e., the fewest number of candidates that we can select while covering $[a,b]$). This can be calculated using a greedy interval cover method over $[a,b]$.
\end{itemize}

Now observe that a supporter interval is contained in $U$ if and only if it is contained in a block of the $U$-decomposition. Thus, we find that the number of cohesive candidates that are SSJR-covered by $U$ is:
\[
c(U) = \sum_{[a,b]\in U} w(a,b).
\]

Thus, we can reformulate the SSJR-maximization problem as the following: select pairwise disjoint and non-adjacent blocks of voters $B = \{[a_1, b_1], \dots, [a_t, b_t]\}$ with $\sum_{[a,b]\in B} \mathrm{cost}(a, b) \leq k$, so as to maximize $\sum_{[a,b]\in B} w(a, b)$.

We consider the following dynamic program to do so. For $j\in[n]$ and $d\in[k]$, take $M(j, d)$ to be the best block weight value attained on some blocks over the voters $1, \dots, j$; such that at most $d$ candidates are selected with $M(0, d) = 0$:

\[
M(j, d) = \max\Big( M(j-1, d),\ \max_{\substack{a \leq j \\ \mathrm{cost}(a,j) \leq d}} M\big(a-1,\, d - \mathrm{cost}(a,j)\big) + w(a,j) \Big).\footnote{Note that this formulation allows for adjacent blocks but this causes no issues as these blocks can be merged to yield a higher score.}
\]

The answer of the DP is $M(n, k)$. This runs in polynomial time, roughly $O(n^2 k)$ after the precomputation of the block weight and cost values. Note that this DP uses at most $k$ candidates, so if less than $k$ are selected, we can simply fill the rest of the committee without changing the result.

\end{proof}

\begin{reptheorem}{thm:fpt}
Assuming we are given the candidate-deletion set $D$, there exists an FPT algorithm to perform SJR maximization with respect to $|D|$.
\end{reptheorem}
\begin{proof}
We take a similar approach to \citet{SornatWX22} who provided an FPT algorithm w.r.t. $|D|$ for computing the $CC$ rule. First, we must guess some subset of the candidate deletion set $D$ that is to be the \emph{pre-elected winners}. Note that we have to consider at most $2^{|D|}$ subsets of $D$ to do so. Second, we remove all candidates in $D$ from the instance so as to make an election in the CI domain. Then for each guessed subset $D'$ (where we only consider guesses where $|D'|\leq k$) we do the following. 
For a guessed subset $D'$, we determine how many candidates are SJR-covered by guessed winners and add this number to the SJR score. We then translate our modified election into an instance of \textsc{Partial Interval Hitting Set} as done in the proof of Theorem~\ref{thm:(S)SJR-max-CI-VI-poly}, except for three differences. First, the candidate intervals from already removed, guessed winners are omitted. Second, the candidate intervals of any remaining candidates that are SJR-covered by guessed winners are set to an empty interval. This prevents us from double-counting candidates that are already covered w.r.t. SJR when we do the maximization. Third, for each cohesive candidate $d\in D\setminus D'$ that is not SJR-covered by any candidate in the guessed set $D'$, we add the single candidate interval $\bigcap_{i\in N_d}(A(i)\setminus D)$ to the instance, i.e., the candidate interval $I^{CI}(d)$ restricted to $C\setminus D$. This is an intersection of intervals, thus also an interval in the CI order.
Finally, we employ the polynomial-time algorithm of \citet{Damaschke17intervals,ChrobakGLN21schedule} to obtain the SJR maximizing committee that supposes $D'$ to be the set of pre-elected winners. We keep track of all committees for each considered $D'$, and at the end, we return the committee with the highest SJR maximization score.

\end{proof}

\section{Supplementary material to \Cref{sec:DR}}
\begin{repobservation}{prop:dr-sat}
Not all instances allow for a committee that satisfies Distinctive Representation. Even on CEI/VEI instances, DR is insatisfiable.
\end{repobservation}
\begin{proof} Any instance that does not have $k$ candidates with size $\lceil n/k\rceil$ serves as a witness for the general case. An example of an instance that also satisfies CEI/VEI domain properties is an instance where all voters approve of all candidates with $k\geq 2$.
\end{proof}

\begin{reptheorem}{thm:dr-hard}
The Distinctive Representation Decision Problem is NP-complete.
\end{reptheorem}
The proof follows by reduction from Exact 3-Cover again.
\begin{proof}
First we show DR is in NP. Consider an instance $I$. A valid witness $W$ satisfies JR and has $|N_c|=\lceil n/k\rceil$ for all $c\in W$. We can verify in polynomial time whether $W$ satisfies JR and we can verify in polynomial time whether each $c\in W$ has $\lceil n/k\rceil$ supporting voters. This shows DR is in NP. 

The proof of hardness follows by a reduction from X3C. Consider any instance $I=(U,\mathcal{S})$ of X3C with $|U|=3q$ for some $q\in\mathbb{N}$ and $\mathcal{S}$ a family of 3-element subsets. We map this to an instance of DR with $3q+4$ voters, $|\mathcal{S}|+1$ candidates, and where $k=q$. To do this, we first copy all $3q$ elements $u_i$ to what we call element-voters $v_i$. Then we add 4 dummy-voters $d_1, d_2,d_3$ and $d_4$. We take $k=q$, so we have $\lceil n/k\rceil = \lceil\frac{3q+4}{q}\rceil=3+\lceil\frac{4}{q}\rceil=4$ for any $q\ge 4$, which we will assume because this limits us to at least 16 voters, but the hardness results are only interesting from somewhat large voter sets anyway. For the candidates, we make one set-candidate $c_i$ for each set $S_i\in\mathcal{S}$ as follows: $c_i=S_i\cup\{d_1\}$ (4 voters). We also make one large candidate $L = U\cup \{d_2,d_3,d_4\}$ ($3q+3$ elements). 

When $I\in$ X3C, there exist $q$ sets $S_i$ that cover all of $U$. In $f(I)$ these $q$ sets correspond to $k$ set-candidates and all have $4=\lceil n/k\rceil$ approving voters. Electing for $W$ precisely these $k$ set-candidates, we cover all voters except for voters $d_2,d_3,d_4$, but these are only three voters and so JR is not violated. Moreover, all elected set-candidates have $\lceil n/k\rceil$ approving voters. Thus, when $I$ is in X3C, $f(I)$ is in DR.

When $I\notin \text{X3C}$, no $q$ sets cover all of U, so in $f(I)$ there exist no $k$ set-candidates that cover all element-voters. But if an element-voter remains uncovered we need to elect $L$ in order to satisfy JR (because otherwise the 4 uncovered voters form a cohesive group). However, $|L|>\lceil n/k\rceil$. It is thus not possible to elect candidates of the right size while at the same time satisfying JR, and so $f(I)\notin$ DR.

This shows that $I\in \text{X3C}\iff f(I)\in$ DR. Our function $f$ runs in polynomial time: we have copied $U$, added four dummy voters, mapped $q$ to $k$, copied all sets $S_i$ and amended them with one additional element, and finally added one set of $3q+3$ elements. This can all be done in polynomial time.
\end{proof}

\section{Supplementary material to \Cref{sec:LocalDR}}
\begin{repproposition}{prop:local-dr-alg}
The LS-DR algorithm satisfies Local DR and runs in polynomial time.
\end{repproposition}
\begin{proof}
The algorithm satisfies Local DR by construction. Regarding the run-time, a first question is whether the while-loop terminates. Consider the objective function $DR$ that maps a committee $W$ to $\sum_{c\in W}||N_c|-\lceil n/k\rceil|$. This function decreases strictly over the run of the algorithm, so the while-loop will terminate in finite time. Furthermore, the objective function value we can only decrease in steps of size $\ge 1$, and for $k\ge 2$, the function $DR$ takes values in $\{0,\ldots,kn\}$. This means at most $kn$ iterations of the while-loop occur. In each iteration, for at most $m^2$ pairs of candidates $(c,c')$ we have to check whether $(W\cup\{c'\})\setminus\{c\}$ satisfies JR (which can be done in polynomial time) and for each pair the improvement in score must be computed, which requires computing $|N_c|$ and $|N_{c'}|$ (this takes $O(n)$ time). Finding the best swap requires comparing all improvements which takes $O(m^2)$ time. Swapping candidates takes constant time. Lastly, a JR satisfying starting committee has to be chosen. This can be done in polynomial time using e.g. the Greedy-CC algorithm. All in all every step takes polynomial time so the algorithm runs in polynomial time.
\end{proof}

\begin{repproposition}{prop:find-dr-committee}
There exist instances where a DR committee exists, but LS-DR may fail to find it.
\end{repproposition}
\begin{proof}
Consider the following instance. Suppose $N=\{v_1,\ldots,v_{24}\}$ ($n=24$) and $k=4$ so that $n/k=6$. There are 7 candidates: $c_1,\ldots c_7$. The approvals are as follows: $N_{c_1}= \{v_1,\ldots, v_4\},\, N_{c_2}\cap N_{c_3}\cap N_{c_4}=\{v_5,\ldots, v_9\}$ and all have one unique voter of $\{v_{10},v_{11},v_{12}\}$. Then $N_{c_5}=\{v_{13},\ldots,v_{{18}}\}$, $N_{c_6}=\{v_{19},\ldots,v_{24}\}$ and $N_{c_7}=N_{c_5}\cup N_{c_6}$. Now any committee $W$ with two candidates from the set $\{c_2,c_3,c_4\}$ and both of the candidates $c_5$ and $c_6$ satisfies DR. After all, all these candidates have support size $n/k$ and no group of more than $n/k$ voters in one candidate remains unrepresented (the only voters that remain unrepresented are one of $\{v_{10},v_{11},v_{12}\}$, and $N_{c_1}$). However, an example of a local optimum is the following: $W=\{c_2,c_3,c_4,c_7\}$. This is clearly suboptimal because $c_7$ is very large, but we cannot swap out $c_7$ and exchange it by $c_5$ and $c_6$, because that requires first replacing any of $c_2,c_3,c_4$ by either $c_5$ or $c_6$, whereas they are the same size, so the LS algorithm will not do that. (Note that this suboptimal committee does actually maximize SJR-score.)
\end{proof}

\begin{definition}[Priceability]\label{def:priceability}
    Suppose that voters have an individual budget $\nicefrac{k}{n}$. A \emph{price system} $\textbf{ps} = (p,\{g_{i}\}_{i\in [n]})$ is a pair where $p$ is the \emph{price} of the candidates and each voter $i\in N$ has a \emph{payment function} $g_{i}: C \rightarrow [0, \nicefrac{k}{n}]$ such that:
 
\begin{itemize}
     \item[\rm{C1}.] If $g_{i}(c) > 0$, then $c\in A(i)$ for all voters $i\in N$.
     \item[\rm{C2}.] $\sum_{c\in C} g_{i}(c) \leq \nicefrac{k}{n}$.
\end{itemize}

\noindent
A price system supports a committee $W$ if all of the following conditions hold:

\begin{itemize}
    \item[\rm{C3}.] For every candidate $c \in W$, $\sum_{i\in N} g_{i}(c) = p$.
    \item[\rm{C4}.] $\sum_{i\in N } g_{i}(c) = 0$ for every candidate $c\notin W$.
    \item[\rm{C5}.] For every candidate $c\notin W$, it holds that: 
    \[
    \sum\limits_{i\in N_c} \bigg(\frac{k}{n}-\sum\limits_{c'\in W} g_{i}(c')\bigg) \leq p
    \]
    (supporters of any unelected candidate do not collectively hold, in terms of their leftover budgets, strictly more in funds than the price of a candidate).
\end{itemize}
\end{definition}

\noindent
If there exists a price system $\textbf{ps} = (p,\{g_{i}\}_{i\in [n]})$ that supports a committee $W$ then we say a committee $W$ is \emph{priceable}.

\begin{reptheorem}{thm:DR-proportionality}
Local DR does not imply EJR or PJR and vice versa. DR does not imply priceability, and priceability does not imply Local DR. Local DR does not imply (S)SJR (even when satisfiable), and (S)SJR does not imply Local DR.
\end{reptheorem}

\begin{proof}
Local DR does not imply EJR nor PJR and vice versa. Consider the following instance, also portrayed in \Cref{fig:DR-proportionality}. We have $n=15$ and the approval sets $N_a=N_b=\{1,2,3,4,5,6\},\,N_c=\{7,8,9\},\,N_d=\{9,10,11\},\,N_e=\{12,13,14\},\,N_f=\{9,10,14\},\,N_g=\{10,14,15\}.$ Suppose $k=5$, then $n/k=3$ and both EJR and PJR require electing both $a$ and $b$. Any three of the other candidates will suffice to fill the committee, because all other candidates are only cohesive, and no group of three voters will be left out when three candidates are elected. But Local DR will never elect both $a$ \emph{and} $b$, but always just one of these. The other candidate is too large and Local DR will swap to another candidate that takes size precisely $n/k=3$. Note that, in this way, Local DR can even cover all voters, by electing e.g. $W=\{a,c,d,e,g\}$.

DR does not imply priceability, and priceability does not imply Local DR: Take $k=2$ and the $5$ voters $A_{1}=A_{2}=A_{3}=\{a,b\}, A_{4} = A_{5} = \{c\}$. Committee $W=\{a,b\}$ provides DR and thus, it provides Local DR. However, it is not priceable. If the price per committee member is set to be strictly higher than $\nicefrac{3}{5}$ then voters $\{1,2,3\}$ cannot afford to fund the committee, but if it is at most $\nicefrac{3}{5}$ then voters $\{4,5\}$ have more in leftover funds than the price, which violates priceability. We can reuse this same election instance to show that in the other direction, priceability does not imply Local DR. Consider the committee $W=\{a,c\}$. This is priceable with a price system where each voter having budget $\nicefrac{2}{5}$ and the candidate price is $\nicefrac{4}{5}$, hence JR is satisfied, but it fails Local DR. To see this, observe that candidate $c$ can be swapped with candidate $b$ to improve the committee's DR-score while JR is preserved since voters 4 and 5 do not form a cohesive group.

A Local DR winning committee is not always an (S)SJR winning committee, and vice versa: Both directions follow from the example in \Cref{prop:monroe-dr}. For both SSJR and SJR we need one of the candidates $c_1,\ldots,c_{\ell}$ but this is not locally optimal for DR. The committee $\{c_{\ell+1},\ldots,c_{\ell+10}\}$ satisfies Local DR but leaves one voter of a cohesive group uncovered so it does not satisfy SSJR nor SJR.
\end{proof}
Note that a complexity argument proves a weaker variant: Unless P = NP, Local DR committees cannot always be (S)SJR-optimal. This is true because LS-DR runs in polynomial time and both the SJR and SSJR decision problems are NP-complete. However, the theorem statement needs the explicit counter-example via the Monroe-instance, because it requires proving that Local DR does not imply (S)SJR axiom satisfiability \emph{on instances where this is possible}, which is not implied by Local DR failing to output (S)SJR-optimal committees \emph{on general instances}.

\begin{figure}[h]\caption{Instance in Theorem~\ref{thm:DR-proportionality}, with approval sets $N_a=N_b=\{1,2,3,4,5,6\},\,N_c=\{7,8,9\},\,N_d=\{9,10,11\},\,N_e=\{12,13,14\},\,N_f=\{9,10,14\},\,N_g=\{10,14,15\}$.}\label{fig:DR-proportionality}
\begin{tikzpicture}[scale=1.7, thick]
\begin{scope}[shift={(0, 0.5)}]
    \draw (0, 0) circle (1);
    \draw (0, 0) circle (0.8);
        
    \fill (-0.45, 0.3) circle (0.07);
    \fill (0, 0.3) circle (0.07);
    \fill (0.45, 0.3) circle (0.07);
        
    \fill (-0.45, -0.3) circle (0.07);
    \fill (0, -0.3) circle (0.07);
    \fill (0.45, -0.3) circle (0.07);
\end{scope}
\begin{scope}[shift={(1.5, 0.5)}]
    \draw[rounded corners=7.1pt] (0,0) rectangle (1.5,0.5);
    \fill (0.25, 0.25) circle (0.07cm);
    \fill (0.75, 0.25) circle (0.07cm);
    \fill (1.25, 0.25) circle (0.07cm);
    \draw[rounded corners=7.1pt] (1,0.05) rectangle (2.5,0.55);
    \fill (1.75, 0.25) circle (0.07cm);
    \fill (2.25, 0.25) circle (0.07cm);
    \draw[rounded corners=7.1pt, decorate=false] (0.25,-0.5) rectangle (1.79,0);
    \fill (0.5, -0.25) circle (0.07cm);
    \fill (1.0, -0.25) circle (0.07cm);
    \fill (1.5, -0.25) circle (0.07cm);
    \fill (2.05, -0.25) circle (0.07cm);
    \draw [rounded corners] (1.1,-0.64)--(2.5,-0.64)--(1.8,0.7)--cycle;
    \draw [rounded corners] (0.75,0.64)--(2.15,0.64)--(1.4,-0.7)--cycle;
\end{scope}

\node[font=\small, minimum width=11pt] at (-0.8,1.43) {$a$};
\node[font=\small] at (-1,1.14) {$b$};
\node[font=\small] at (1.51,1.12) {$c$};
\node[font=\small] at (4.08,1.13) {$d$};
\node[font=\small] at (1.74,-0.14) {$e$};
\node[font=\small] at (2.35,1.36) {$f$};
\node[font=\small] at (4.07,0.07) {$g$};
\node[font=\tiny] at (-0.35,0.95) {$1$};
\node[font=\tiny] at (0.11,0.95) {$2$};
\node[font=\tiny] at (0.55,0.95) {$3$};
\node[font=\tiny] at (-0.35,0.3) {$4$};
\node[font=\tiny] at (0.11,0.3) {$5$};
\node[font=\tiny] at (0.55,0.3) {$6$};
\node[font=\tiny] at (1.85,0.85) {$7$};
\node[font=\tiny] at (2.35,0.85) {$8$};
\node[font=\tiny] at (2.85,0.85) {$9$};
\node[font=\tiny] at (3.4,0.85) {$10$};
\node[font=\tiny] at (3.9,0.85) {$11$};
\node[font=\tiny] at (2.15,0.35) {$12$};
\node[font=\tiny] at (2.65,0.35) {$13$};
\node[font=\tiny] at (3.15,0.35) {$14$};
\node[font=\tiny] at (3.7,0.35) {$15$};
\end{tikzpicture}
\end{figure}

\begin{repproposition}{prop:monroe-dr}
An optimal Monroe committee need not satisfy (Local) DR, and vice versa.
\end{repproposition}
\begin{proof}
Suppose $n=100,\,k=10$ and so $n/k=10$. There are $\ell\ge10$ candidates $c_1,\ldots,c_\ell$ that all voters approve of. Moreover, there are 10 candidates that \emph{nearly} partition the voters into groups of equal size. Candidates $c_{\ell+1},\ldots,c_{\ell+8}$ all contain 10 voters and, writing these candidates as sets of voters, these eight candidates are disjoint of one another. Then there are $c_{\ell+9}$ and $c_{\ell+10}$, which also contain 10 voters, and are disjoint from candidates $c_{\ell+1},\ldots,c_{\ell+8}$, but the pair overlaps by 1 voter, which means there is 1 voter left that is not in any of $c_{\ell+1},\ldots,c_{\ell+10}$. Now, an optimal DR assignment is available: elect $c_{\ell+1},\ldots,c_{\ell+10}$. These all have $n/k=10$ supporting voters and only one voter remains unrepresented, so JR is satisfied. Any committee $W$ that contains a candidate from $\{c_1,\ldots,c_{\ell}\}$ is not locally optimal, because we can swap out this big candidate and replace it for one of size $n/k$ while maintaining JR. But an optimal Monroe assignment will not accept this one voter remaining uncovered and will pick at least one candidate from $c_1,\ldots,c_{\ell}$ in order to prevent that, so that an optimal Monroe committee for this instance has an assignment where all voters are mapped to a candidate they approve of. This instance shows that Monroe and Local DR are distinct.
\end{proof}

\begin{proposition}
   An optimal Chamberlin Courant winning committee need not satisfy (Local) DR, and vice versa.
\end{proposition}
\begin{proof}
    Consider the example in the proof showing that Monroe need not satisfy Local DR (Proposition~\ref{prop:monroe-dr}). Observe that the same holds for a Chamberlin-Courant winning committee; all voters can be covered by taking a candidate from $\{c_1,\ldots,c_{\ell}\}$, so that is what CC will do. So that same instance shows that CC is also distinct from Local DR.
\end{proof}

\begin{reptheorem}{thm:DR-BJR}
\begin{itemize}
        \item
        A committee $W$ that satisfies DR (when satisfiable) needn't satisfy BJR for approval ballots, and vice versa.
        \item
        A committee $W$ that satisfies Local DR needn't satisfy BJR for approval ballots, and vice versa.
\end{itemize}
\end{reptheorem}

\begin{proof}

\begin{itemize}
    \item The following simple instance shows that BJR for approval ballots does not imply DR-when-possible. Take $n=4$ and $k=2$ so that $n/k=2$. As candidates we have $c_1=\{v_1,v_2,v_3,v_4\}$ and $c_2=\{v_1,v_2\},\,c_3=\{v_3,v_4\}$. Then $W=\{c_1,c_2\}$ satisfies BJR for approval ballots, with assignment function $\omega$ that assigns voters $v_1$ and $v_2$ to candidate $c_2$ and assigns voters $v_3$ and $v_4$ to candidate $c_1$. But this committee does not satisfy DR-when-satisfiable, because candidate $c_1$ has four approving voters, whereas candidate $c_3$ contains only $n/k=2$ voters and replacing $c_1$ by $c_3$ would still cover all voters, thus satisfy JR.
    
    The following instance shows that DR-when-possible does not imply BJR for approval ballots. Take $n=19$ and $k=9$ so that $\lceil n/k\rceil=3$ and $\lfloor n/k\rfloor =2$ and consider first the four candidates $c_1=\{v_1,v_2,v_3\},\,c_2=\{v_1,v_3,v_4\},\,c_3=\{v_1,v_2,v_4\},\,c_4=\{v_2,v_3,v_4\}$. Beyond these four candidates, for each of the $\binom{15}{3}$ sets of 3 out of the remaining 15 voters, there is one candidate that has exactly these three voters approving them. We call five of them that partition these 15 voters, candidates $c_5,\ldots,c_9$. Then $W=\{c_1,\ldots,c_9\}$ satisfies DR; all candidates have $\lceil n/k\rceil =3$ approving voters and all voters have a representative in $W$, so JR is also satisfied. But there exist no assignment function $\omega$ that would ensure satisfaction of BJR: Candidates $c_1,\ldots c_4$ need at least $4\cdot\lfloor n/k\rfloor =8$ voters assigned to them, so that requires assigning 4 voters from the set $\{v_5,\ldots,v_{19}\}$. But for any four we pick, there exists a subset of three of these voters with a non-elected candidate that these three all approve of. This means these three voters are not assigned to a candidate they approve of, whereas there does exist a candidate they all jointly approve of. This means the committee does not satisfy BJR.
    \item The following instance shows that Local DR and BJR for approval ballots are logically distinct: we have $n=12$ voters and $k=6$, so that $n/k=2$. We have 7 candidates $c_1,\ldots,c_7$. Candidates $c_1$ and $c_2$ are approved by all voters. Candidates $c_3,\ldots,c_7$ all have 2 approving voters, but, though all candidates are distinct, they overlap slightly, so that a total of 8 voters are covered by these 5 candidates (see \Cref{fig:prop:LDR-BJR}). That means there are 4 voters left that are covered only by candidates $c_1$ and $c_2$. Now a BJR committee must elect both $c_1$ and $c_2$, because otherwise, for all assignments $\omega$, there will be at least 2 voters that cannot be assigned to a candidate they approve of, whereas such a candidate exists. But a Local DR winning committee will always swap out of such a BJR winner, because JR is still satisfied when we elect only $c_1$ \emph{or} $c_2$, and all 5 other candidates are still cohesive, but smaller.
\end{itemize}
\end{proof}

\begin{figure}[h]\caption{Instance in Theorem~\ref{thm:DR-BJR}, where there are 7 candidates and 12 voters}\label{fig:prop:LDR-BJR}
\begin{tikzpicture}[thick, scale=0.4]
    \draw (0,0) circle (4.5);
    \draw (0,0) circle (4.1);

    \def\rsmall{1.0}    
    \def\rdot{0.15}     

    \begin{scope}[shift={(0.4, 1.8)}]
        \draw (0,0) circle (\rsmall);
        \fill (0, 0.35) circle (\rdot);
        \fill (0,-0.35) circle (\rdot);
    \end{scope}

    \begin{scope}[shift={(-2.0, 0.1)}]
        \draw (0, 0.65) circle (\rsmall);  
        \draw (0,-0.65) circle (\rsmall);  
        \fill (0, 1.1) circle (\rdot);     
        \fill (0, 0) circle (\rdot);       
        \fill (0,-1.1) circle (\rdot);     
    \end{scope}

    \begin{scope}[shift={(10pt,-54pt)}]
        \draw (0, 0.65) circle (\rsmall);  
        \draw (0,-0.65) circle (\rsmall);  
        \fill (0, 1.1) circle (\rdot);     
        \fill (0, 0) circle (\rdot);       
        \fill (0,-1.1) circle (\rdot);     
    \end{scope}

    \begin{scope}[shift={(2.3, 0)}]
        \fill (-0.3, 0.3) circle (\rdot);  
        \fill ( 0.3, 0.3) circle (\rdot);  
        \fill (-0.3,-0.3) circle (\rdot);  
        \fill ( 0.3,-0.3) circle (\rdot);  
    \end{scope}
\node[font=\small] at (-3.67,3.64) {$c_1$};
\node[font=\small] at (-3.67,2.83) {$c_2$};
\node[font=\small] at (-3.45,0.35) {$c_3$};
\node[font=\small] at (-3.45,-0.84) {$c_4$};
\node[font=\small] at (1.83,1.2) {$c_5$};
\node[font=\small] at (1.83,-1.69) {$c_6$};
\node[font=\small] at (1.83,-3.02) {$c_7$};
\end{tikzpicture}
\end{figure}
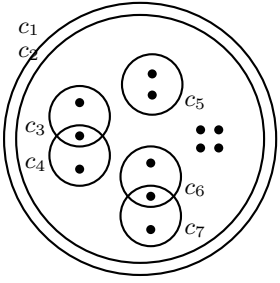

\section{Supplementary material to Section~\ref{sec:exp}}\label{app:exp}

\subsection{(S)SJR}
The table below shows the rules' results with respect to the (S)SJR-scores returned for the two real-world datasets. A first observation is that all rules do really well in general. 
When we take a closer look, we see that SJR and SSJR do perform particularly well on their own measure, and (generally) outperform all other rules, with the exception of the simplified CNASH rule which even beats SSJR on their own measure (averaging over Situ and SCOTUS data). Another interesting note is that SJR performs relatively badly in terms of SSJR-score, and vice versa. This shows these rules, though theoretically optimal for their own optimization measure, and conceptually similar, are indeed really tailored towards their own optimization goal.  

\begin{table}[h]
 \centering
 \begin{tabular}{|c|c|c|c|}
 \hline
& \multicolumn{1}{c|}{} & \multicolumn{1}{c|}{Situ} & \multicolumn{1}{c|}{SCOTUS} \\
 \hline
 \parbox[t]{2mm}{\multirow{6}{*}{\rotatebox[origin=c]{90}{SSJR}}} & Greedy SSJR & $1.0$ & $0.949$\\
 & Greedy SJR & $1.0$ & $0.877$\\
 & Greedy CC & $0.914$ & $0.953$\\
 & MES & $0.987$ & $0.905$\\
& CNASH & $1.0$ & $0.986$\\
& LS-DR & $0.996$ & $0.795$\\
 \hline
 \parbox[t]{2mm}{\multirow{6}{*}{\rotatebox[origin=c]{90}{SJR}}} & Greedy SSJR & $1.0$ & $0.877$\\
 & Greedy SJR & $1.0$ & $0.994$\\
 & Greedy CC & $0.9$ & $0.974$\\
  & MES & $0.987$ & $0.988$\\
& CNASH & $0.989$ & $0.981$\\
& LS-DR & $0.998$ & $0.977$\\
 \hline
 \end{tabular}
 \caption{Table with average (S)SJR scores for the rules on election datasets. Each cell is the average taken over profiles in the dataset with committees sizes $k\in\{2,\ldots,7\}$.}
\label{table:exp_scores}
 \end{table}
 
\subsection{Implementation details}

The following sections provides details on the ILP encodings of the algorithms used for finding committees that satisfy SSJR and SJR. We implemented these in \emph{Python} using the \emph{Gurobi} solver.

\subsubsection{Finding SJR committees}
For this check, we wish to check for every cohesive group $N'$, if there is a committee member that all members of $N'$ approve of. Since there are a large number of cohesive groups, we use \emph{constraint generation} technique. Specifically, for a committee $W$ that the solver poses as a possible solution, we use a verification oracle (implemented as an ILP) to check if there is a cohesive group that causes a failure of SJR, and we add a constraint to this ILP that ensures that this group must collectively agree on a committee member.

\paragraph{Variables:}
We use the following variables:

\begin{itemize}
    \item $y_c \in \{0, 1\}$ is a binary variable for every candidate $c \in C$, where $y_c = 1$ if candidate $c\in W$, and $0$ otherwise.
\end{itemize}


\paragraph{Objective:} Since this is a feasibility problem, we set the objective to $0$.

\paragraph{Constraints}
We start the following constraint:
\begin{itemize}
    \item Choose $k$ candidates for the committee: 
    \[
    \sum_{c \in C} y_c = k
    \]
\end{itemize}

\paragraph{Procedure:}
\begin{enumerate}
    \item Take a committee $W$ chosen by the solver.
    \item We use another ILP (detailed in the next subsection titled ``SJR verification") as an oracle to check if $W$ violates SJR. 
    \item If SJR is violated, the oracle returns a cohesive group $N'$ that is witness to this SJR violation. We add the following constraint and ask the solver to try again with a new committee. $$\sum_{c \in \bigcap_{i \in N'} A(i)} y_c \geq 1$$
    \item For a current committee $W$ satisfying all current constraints, if the oracle cannot find any more SJR violations, then it returns $W$. If all committees have been tested and all violate SJR, then there is no SJR committee, so the algorithm indicates infeasibility. 
\end{enumerate}

\subsubsection{SJR verification}
Given a committee $W$, this ILP is used to find a violation of SJR. Specifically, it is used to search for a cohesive group $N'$ such that no candidate in $W$ is approved by all members of $N'$. 

\paragraph{Variables:} 
We use the following variables:
\begin{itemize}
    \item $z_d \in \{0, 1\}$ for every candidate $d \in C \setminus W$. It equals $1$ if candidate $d$ is the candidate that our counterexample cohesive group $N'$ collectively approves of.
    \item $v_i \in \{0, 1\}$ for every voter $i \in N$. It equals $1$ if voter $i \in N'$.
\end{itemize}

\paragraph{Objective}
Maximize the size of the violating cohesive group $N'$:
\[
\max\sum_{i \in N'} v_i
\]

We set this particular objective since we will add the chosen cohesive group $N'$ (if it is indeed a counterexample) to the set of cohesive groups to be satisfied, so we aim to pick the largest such group.

\paragraph{Constraints}
We have the following constraints:
\begin{itemize}
    \item Only one candidate is considered:
    \[
    \sum_{d \in C \setminus W} z_d = 1
    \]

\item If voter $i$ does not approve of any $d \in C \setminus W$, then $v_i = 0$, and so voter $i$ is not a part of $N'$: 
\[
v_i \leq \sum_{d \in C \setminus W \,:\, d \in A(i)} z_d \quad \forall i \in N
\]

\item For every committee member $c \in W$, some voter in $N'$ disapproves of $c$:
\[
\sum_{i \in N \,:\, c \notin A(i)} v_i \geq 1 \quad \forall c \in W
\]

\item The group must be large enough: $$\sum_{i \in N} v_i \geq \frac{n}{k}$$.
\end{itemize}

\noindent
If the solver finds a feasible solution then it has found a SJR-violating cohesive group. But, if no such group exists, then it indicates infeasibility which means that SJR is satisfied by $W$.

\subsubsection{Finding SSJR committees}
Here, we first identify all of the voters that approve of some cohesive candidate, and amongst these voters, if any are unrepresented by the committee $W$, this represents a violation of SSJR.

\paragraph{Variables:} 
We use the following variables:
\begin{itemize}
    \item $y_c \in \{0, 1\}$ is a binary variable for every candidate $c \in C$, where $y_c = 1$ if candidate $c\in W$, and $0$ otherwise.
\end{itemize}
We enumerate the set of voters who approve of a cohesive candidate: \[
E_{N} = \{v_i\in N\mid \exists c\in C \text{ s.t. } c\in A(i), |N_{c}|\geq \nicefrac{n}{k}\}.
\] 

\noindent
These are the voters that could cause a violation of SSJR.

\paragraph{Objective:} Since this is a feasibility problem, we set the objective to $0$.

\paragraph{Constraints}
\begin{itemize}
    \item Choose $k$ candidates for the committee: 
    $$\sum_{c \in C} y_c = k$$ 
    \item Every potentially violating voter must be represented: \[ \sum_{c \in A(i)} y_c \geq 1 \quad \forall i \in E_{N}\]
\end{itemize}

\subsection{Additional Plots}

This section contains plots that were omitted from the main text.

\begin{figure}[h]
\centering
\begin{subfigure}[b]{0.5\textwidth}
        \centering
        \includegraphics[width=\linewidth]{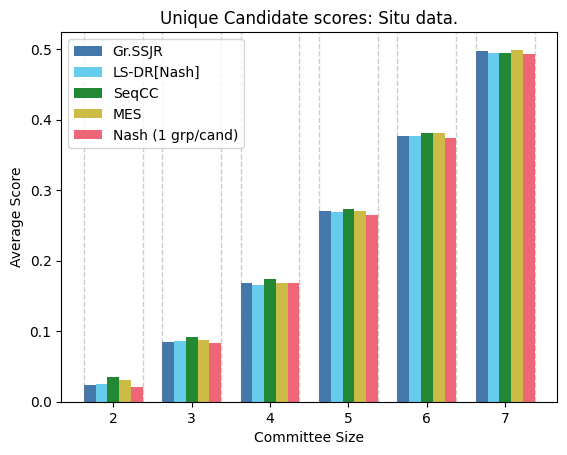}
    \end{subfigure}~\begin{subfigure}[b]{0.5\textwidth}
        \centering
        \includegraphics[width=\linewidth]{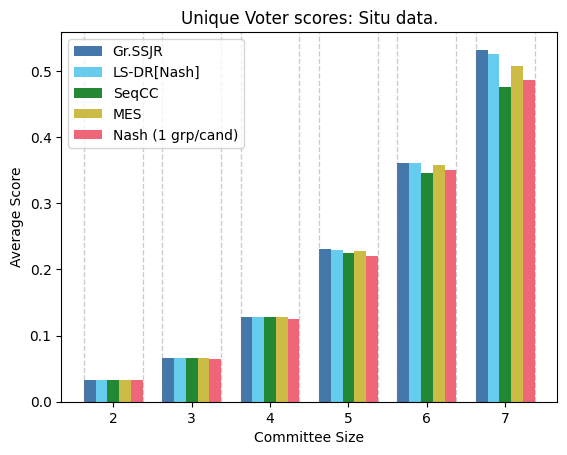}
    \end{subfigure}
\caption{Bar plots for the unique candidate and unique voter scores of the rules on the Situ dataset.}
\label{fig:uniq_cand_vote_situ}
\end{figure}

\begin{figure}[h]
\centering
\includegraphics[width=0.5\textwidth]{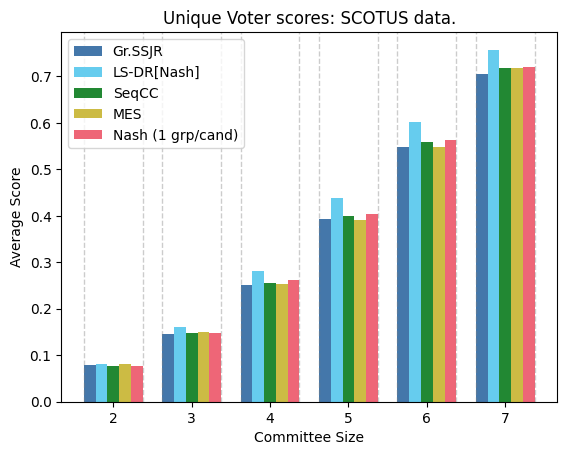}
\caption{Bar plot for the unique vote scores of the rules on the SCOTUS dataset}
\label{fig:uniq_vote_scotus}
\end{figure}

\begin{figure}[h]
\centering
\begin{subfigure}[b]{0.5\textwidth}
        \centering
        \includegraphics[width=\linewidth]{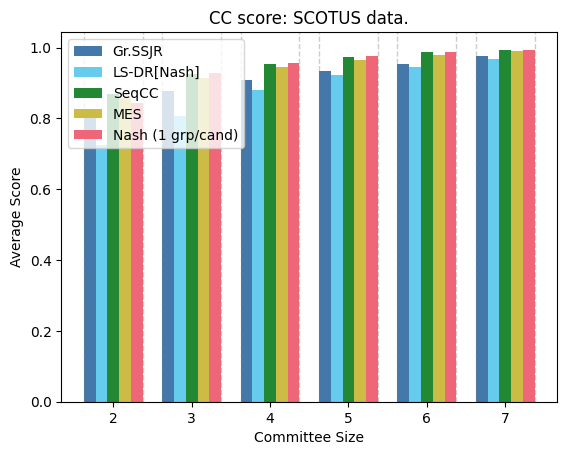}
    \end{subfigure}~\begin{subfigure}[b]{0.5\textwidth}
        \centering
        \includegraphics[width=\linewidth]{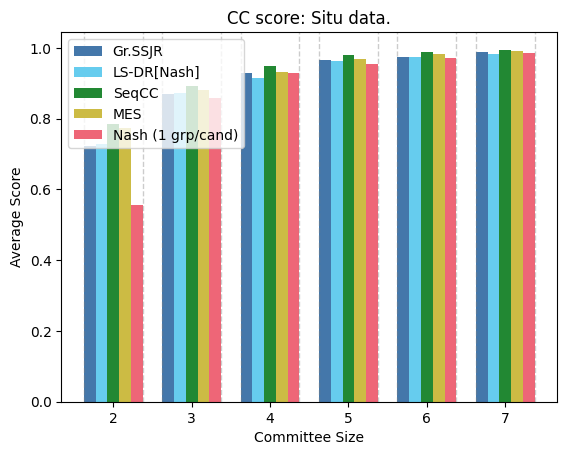}
    \end{subfigure}
\caption{Bar plots for the CC scores of the rules on the SCOTUS and Situ datasets.}
\label{fig:cc_scores}
\end{figure}

\begin{figure}[h]
\centering
\begin{subfigure}[b]{0.315\textwidth}
        \centering
        \includegraphics[width=\linewidth]{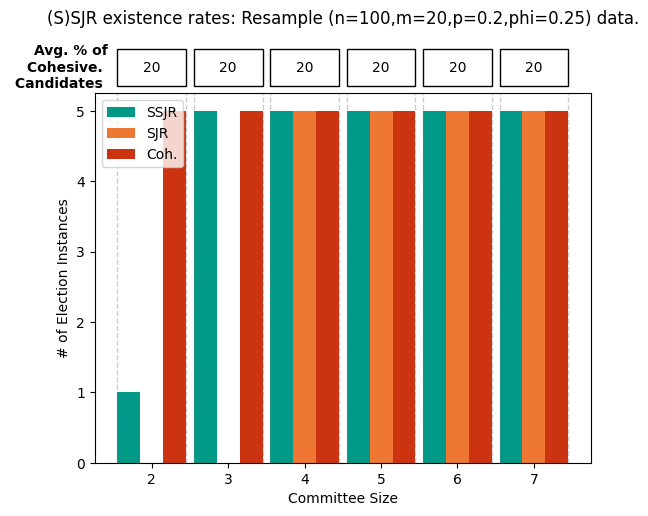}
    \end{subfigure}~\begin{subfigure}[b]{0.315\textwidth}
        \centering
        \includegraphics[width=\linewidth]{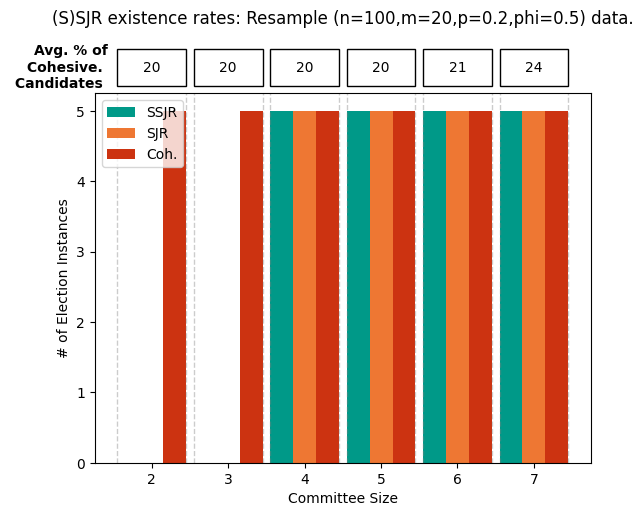}
    \end{subfigure}
    ~\begin{subfigure}[b]{0.315\textwidth}
        \centering
        \includegraphics[width=\linewidth]{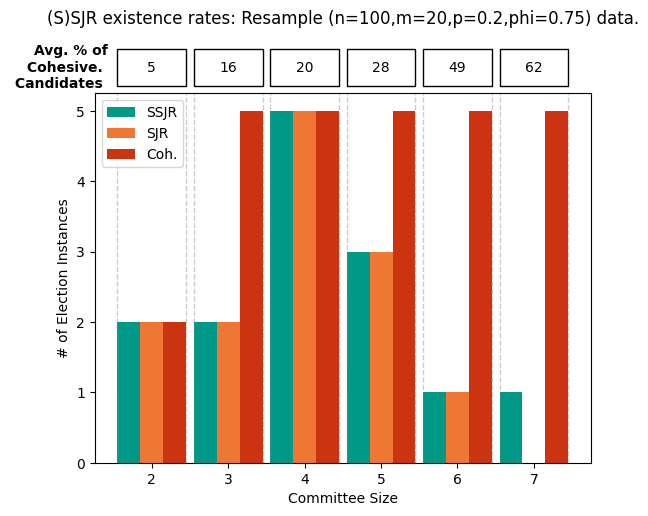}
    \end{subfigure}
\caption{Bar plots for the (S)SJR existence rates of the rules on the synthetic datasets.}
\label{fig:(s)sjr_syn_scores}
\end{figure}

\begin{figure}[h]
\centering
\begin{subfigure}[b]{0.315\textwidth}
        \centering
        \includegraphics[width=\linewidth]{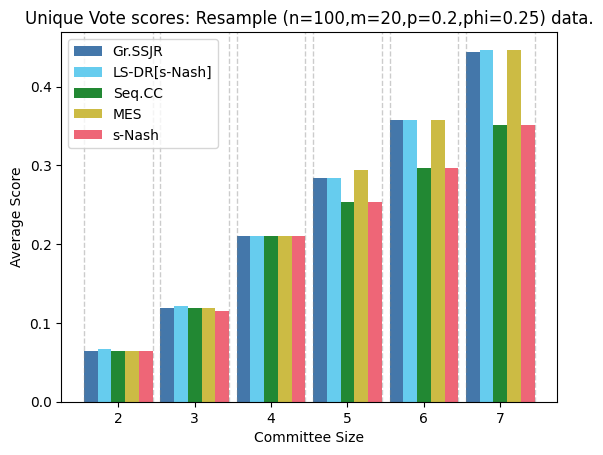}
    \end{subfigure}~\begin{subfigure}[b]{0.315\textwidth}
        \centering
        \includegraphics[width=\linewidth]{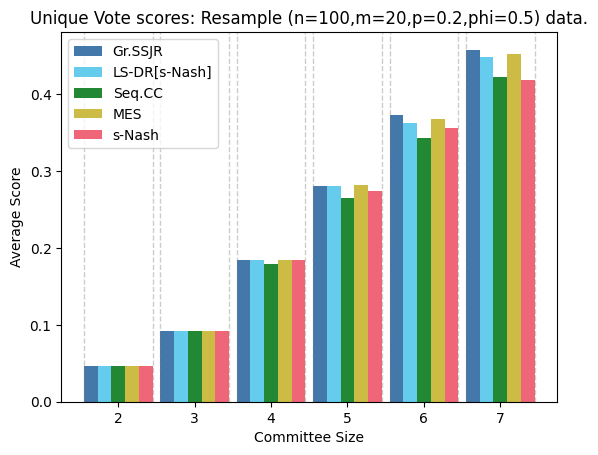}
    \end{subfigure}
    ~\begin{subfigure}[b]{0.315\textwidth}
        \centering
        \includegraphics[width=\linewidth]{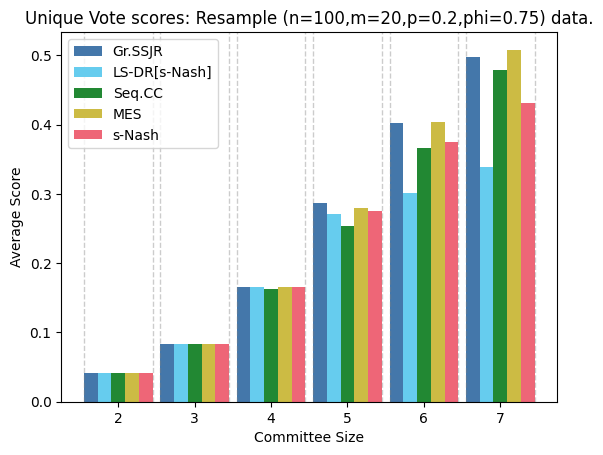}
    \end{subfigure}
\caption{Bar plots for the unique voter-scores of the rules on the synthetic dataset.}
\label{fig:uniq_voter_syn_scores}
\end{figure}

\begin{figure}[h]
\centering
\begin{subfigure}[b]{0.315\textwidth}
        \centering
        \includegraphics[width=\linewidth]{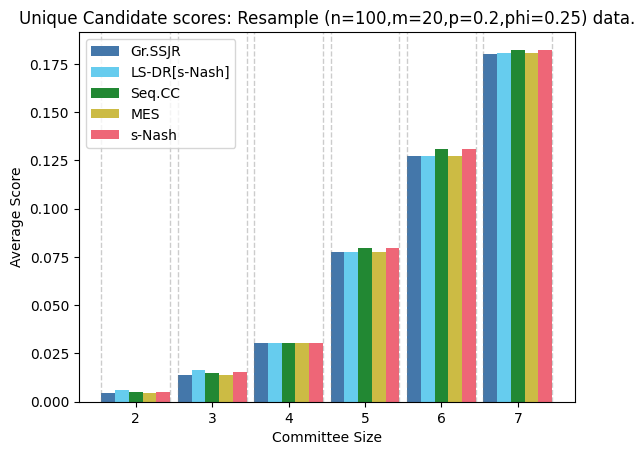}
    \end{subfigure}~\begin{subfigure}[b]{0.315\textwidth}
        \centering
        \includegraphics[width=\linewidth]{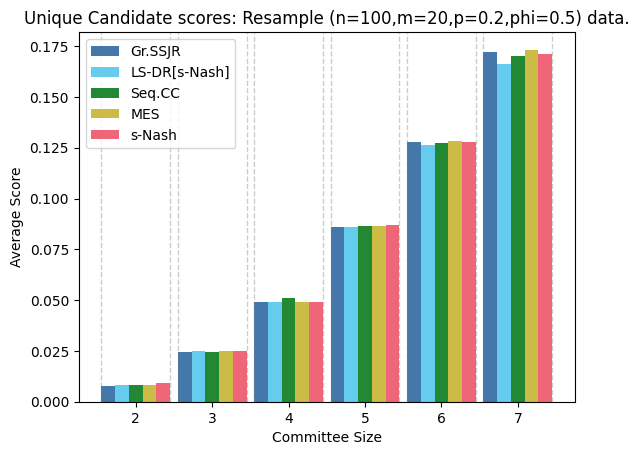}
    \end{subfigure}
    ~\begin{subfigure}[b]{0.315\textwidth}
        \centering
        \includegraphics[width=\linewidth]{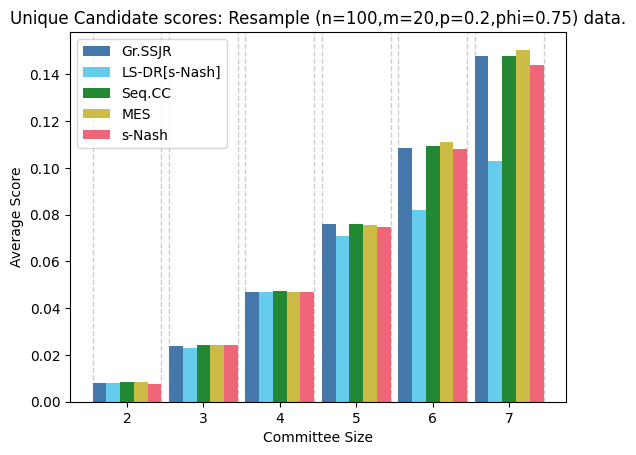}
    \end{subfigure}
\caption{Bar plots for the unique candidate-scores of the rules on the synthetic dataset.}
\label{fig:uniq_cand_syn_scores}
\end{figure}

\begin{figure}[h]
\centering
\begin{subfigure}[b]{0.315\textwidth}
        \centering
        \includegraphics[width=\linewidth]{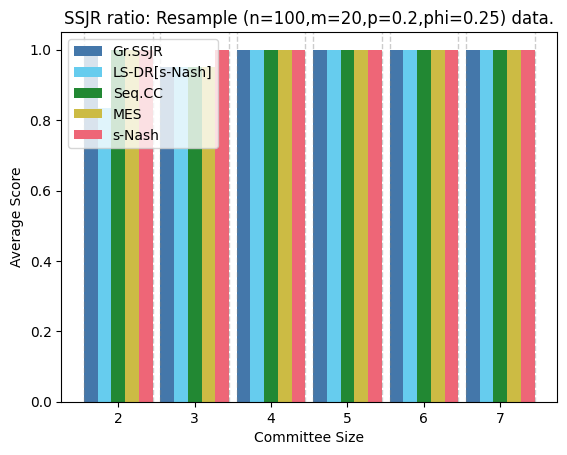}
    \end{subfigure}~\begin{subfigure}[b]{0.315\textwidth}
        \centering
        \includegraphics[width=\linewidth]{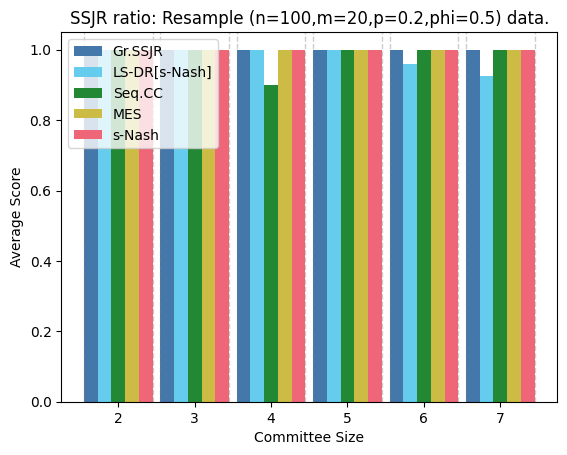}
    \end{subfigure}
    ~\begin{subfigure}[b]{0.315\textwidth}
        \centering
        \includegraphics[width=\linewidth]{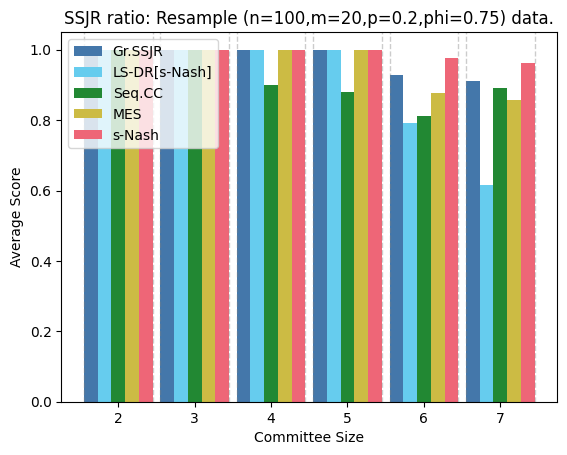}
    \end{subfigure}
\caption{Bar plots for the SSJR scores of the rules on the synthetic dataset.}
\label{fig:ssjr_ratio_syn}
\end{figure}

\begin{figure}[h]
\centering
\begin{subfigure}[b]{0.315\textwidth}
        \centering
        \includegraphics[width=\linewidth]{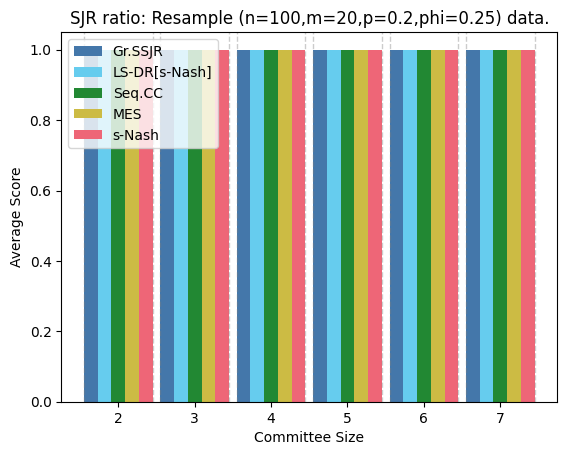}
    \end{subfigure}~\begin{subfigure}[b]{0.315\textwidth}
        \centering
        \includegraphics[width=\linewidth]{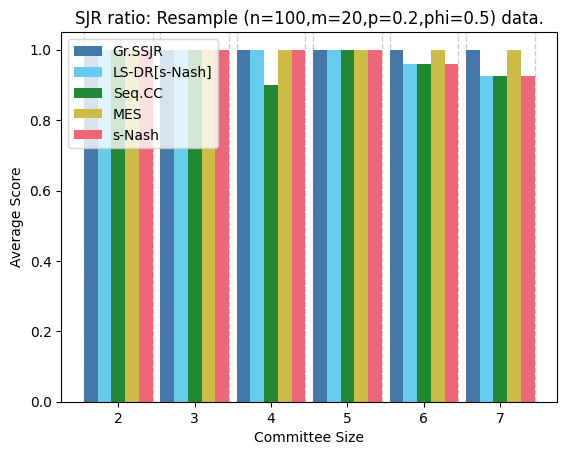}
    \end{subfigure}
    ~\begin{subfigure}[b]{0.315\textwidth}
        \centering
        \includegraphics[width=\linewidth]{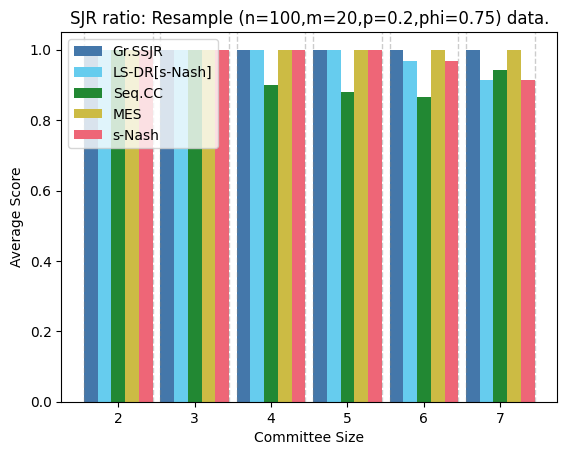}
    \end{subfigure}
\caption{Bar plots for the SJR scores of the rules on the synthetic dataset.}
\label{fig:sjr_ratio_syn}
\end{figure}

\begin{figure}[h]
\centering
\begin{subfigure}[b]{0.315\textwidth}
        \centering
        \includegraphics[width=\linewidth]{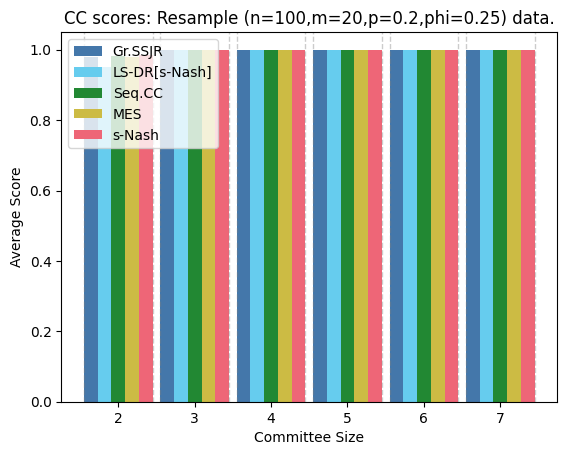}
    \end{subfigure}~\begin{subfigure}[b]{0.315\textwidth}
        \centering
        \includegraphics[width=\linewidth]{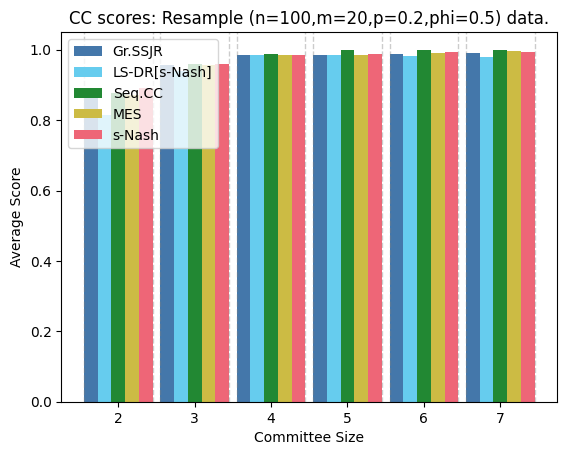}
    \end{subfigure}
    ~\begin{subfigure}[b]{0.315\textwidth}
        \centering
        \includegraphics[width=\linewidth]{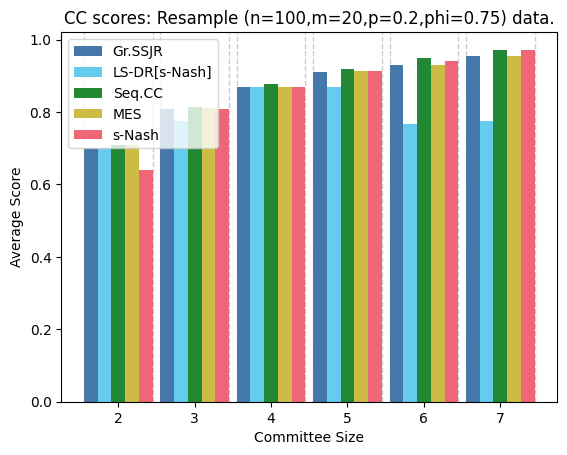}
    \end{subfigure}
\caption{Bar plots for the CC scores of the rules on the synthetic dataset.}
\end{figure}

\clearpage

\end{document}